\pdfoutput=1

\documentclass[11pt, dvipsnames]{article}

\usepackage{natbib}

\usepackage{amsthm}
\usepackage{amsmath}
\usepackage{refcount}
\usepackage{fullpage}
\usepackage{paralist}
\usepackage{appendix}
\usepackage{graphicx}
\usepackage{bm}
\usepackage{bbold}
\usepackage{tikzit} 

\tikzstyle{green dot}=[fill=green, draw=black, shape=circle]
\tikzstyle{red dot}=[fill=red, draw=black, shape=circle]
\tikzstyle{blue dot}=[fill=blue, draw=black, shape=circle]
\tikzstyle{purpule dot}=[fill={rgb,255: red,98; green,0; blue,98}, draw=black, shape=circle]

\tikzstyle{green edge}=[-, draw=green, fill=none]
\tikzstyle{red edge}=[-, draw=red]
\tikzstyle{blue edge}=[-, draw=blue]
\tikzstyle{purpule edge}=[-, draw={rgb,255: red,98; green,0; blue,98}]
\tikzstyle{special edge}=[draw=black, {|-|}]
\tikzstyle{special blue edge}=[draw=blue, {|-|}]

\usepackage{algpseudocode}
\usepackage{verbatim}

\renewcommand{\phi}{\varphi}

\usepackage{todonotes}

\newcommand{\hide}[1]{ }

\renewcommand{\mathbf}{\bm}

\theoremstyle{plain}

\renewcommand{\include}{\input}

\usepackage{amsfonts}

\usepackage{tikz}
\usetikzlibrary{decorations.pathreplacing,angles,quotes}
\usetikzlibrary{arrows}
\usepackage{float}

\usetikzlibrary{patterns}

\usepackage{libertine}

\usepackage{geometry}
\usepackage{enumitem}
\PassOptionsToPackage{vlined, boxruled}{algorithm2e}
\usepackage{algorithm2e}
\LinesNumbered

\usepackage{mleftright}
\usepackage{relsize}

\usepackage{caption}
\usepackage{subcaption}

\usepackage{bbm}

\usepackage{amsmath}
\usepackage{amsthm}
\usepackage{graphicx}
\usepackage{tabularx}
\newcolumntype{Y}{>{\centering\arraybackslash}X}

\usepackage{array}
\usepackage{soul}

\usepackage{mdframed}
\usepackage{xparse}
\usepackage{hyperref}
\hypersetup{colorlinks,linkcolor=blue,citecolor=blue,urlcolor=blue}

\usepackage[capitalise, noabbrev]{cleveref}

\usepackage{tcolorbox}
\usepackage{tikz,lipsum}
\usepackage{newtxmath}
\tcbuselibrary{skins,breakable}

\theoremstyle{plain}
\newtheorem{thm1}{Theorem}[section]
\theoremstyle{remark}
\newtheorem{pthm1}[thm1]{Theorem}
\theoremstyle{plain}
\newtheorem{lem1}[thm1]{Lemma}
\theoremstyle{plain}
\newtheorem{obs1}[thm1]{Observation}
\theoremstyle{plain}
\newtheorem{inv1}[thm1]{Invariant}
\theoremstyle{plain}
\newtheorem{cor1}[thm1]{Corollary}
\theoremstyle{definition}
\newtheorem{defn1}[thm1]{Definition}
\theoremstyle{plain}
\newtheorem{fact1}[thm1]{Fact}
\theoremstyle{remark}
\newtheorem{rem1}[thm1]{Remark}
\theoremstyle{plain}
\newtheorem{prop1}[thm1]{Proposition}
\theoremstyle{plain}
\newtheorem{asmp1}[thm1]{Assumption}

\ifx\proof\undefined
\newenvironment{proof}[1][\protect\proofname]{\par
\normalfont\topsep6\p@\@plus6\p@\relax
\trivlist
\itemindent\parindent
\item[\hskip\labelsep\scshape #1]\ignorespaces
}{%
\endtrivlist\@endpefalse
}
\providecommand{\proofname}{Proof}
\fi

\def\special{1}
\def\specialproof{1}
\def\specialdefinition{1}
\def\specialremark{1}
 \def\highlight{0}
 \def\sidebar{1}

\def\colorequations{0}

\newenvironment{theorem}[1][]{%
\begin{thm1}[#1]%
}{\end{thm1}%
}
\newenvironment{lemma}[1][]{%
\begin{lem1}[#1]%
}{\end{lem1}%
}
\newenvironment{observation}[1][]{%
\begin{obs1}[#1]%
}{\end{obs1}%
}
\newenvironment{inv}[1][]{%
\begin{inv1}[#1]%
}{\end{inv1}%
}
\newenvironment{fact}[1][]{%
\begin{fact1}[#1]%
}{\end{fact1}%
}
\newenvironment{remark}[1][]{%
\begin{rem1}[#1]%
}{\end{rem1}%
}
\newenvironment{pthm}[1][]{%
\begin{pthm1}[#1]%
}{\end{pthm1}%
}
\newenvironment{corollary}[1][]{%
\begin{cor1}[#1]%
}{\end{cor1}%
}
\newenvironment{definition}[1][]{%
\begin{defn1}[#1]%
}{\end{defn1}%
}
\newenvironment{asmp}[1][]{%
\begin{asmp1}[#1]%
}{\end{asmp1}%
}
\newenvironment{proposition}[1][]{%
\begin{prop1}[#1]%
}{\end{prop1}
}%

\if 1
    \if 1\highlight
        \renewenvironment{theorem}[1][]{%
        \begin{mdframed}[nobreak=false,backgroundcolor=Aquamarine!60]\begin{thm1}[#1]%
        }{\end{thm1}\end{mdframed}%
        }
        \renewenvironment{lemma}[1][]{%
        \begin{mdframed}[nobreak=false,backgroundcolor=YellowGreen!60]\begin{lem1}[#1]%
        }{\end{lem1}\end{mdframed}%
        }
        \renewenvironment{observation}[1][]{%
        \begin{mdframed}[nobreak=false,backgroundcolor=Salmon!60]\begin{obs1}[#1]%
        }{\end{obs1}\end{mdframed}%
        }

        \renewenvironment{corollary}[1][]{%
        \begin{mdframed}[backgroundcolor=Mulberry!60]\begin{cor1}[#1]%
        }{\end{cor1}\end{mdframed}%
        }
    \fi
    \if 1\sidebar
        \tcolorboxenvironment{theorem}{blanker,breakable,left=4mm, before skip=10pt,after skip=10pt, borderline west={2mm}{0pt}{Aquamarine}}
        \tcolorboxenvironment{lemma}{blanker,breakable,left=4mm, before skip=10pt,after skip=10pt, borderline west={2mm}{0pt}{YellowGreen}}
        \tcolorboxenvironment{observation}{blanker,breakable,left=4mm, before skip=10pt,after skip=10pt, borderline west={2mm}{0pt}{Salmon}}
        \tcolorboxenvironment{inv}{blanker,breakable,left=4mm, before skip=10pt,after skip=10pt, borderline west={2mm}{0pt}{Salmon}}
        \tcolorboxenvironment{fact}{blanker,breakable,left=4mm, before skip=10pt,after skip=10pt, borderline west={2mm}{0pt}{Salmon}}
        \tcolorboxenvironment{pthm}{blanker,breakable,left=4mm, before skip=10pt,after skip=10pt, borderline west={2mm}{0pt}{Salmon}}
        \tcolorboxenvironment{proposition}{blanker,breakable,left=4mm, before skip=10pt,after skip=10pt, borderline west={2mm}{0pt}{Goldenrod}}
        \tcolorboxenvironment{corollary}{blanker,breakable,left=4mm, before skip=10pt,after skip=10pt, borderline west={2mm}{0pt}{Mulberry}}
        \tcolorboxenvironment{asmp}{blanker,breakable,left=4mm, before skip=10pt,after skip=10pt, borderline west={2mm}{0pt}{Goldenrod}}
    \fi
\fi

\if 1\specialproof
    \if 1\highlight
        \expandafter\let\expandafter\oldproof\csname\string\proof\endcsname
        \let\oldendproof\endproof
        \renewenvironment{proof}[1][\proofname]{%
        \begin{mdframed}[nobreak=false,backgroundcolor=lightgray!60]\oldproof[#1]%
        }{\oldendproof\end{mdframed}}
    \else
        \if 1\sidebar
        \tcolorboxenvironment{proof}{
        blanker,breakable,left=4mm,
        before skip=10pt,after skip=10pt,
        borderline west={2mm}{0pt}{lightgray}}
        \fi
    \fi
\fi

\if 1\specialdefinition
    \if 1\highlight
        \renewenvironment{definition}[1][]{%
        \begin{mdframed}[innerbottommargin=0.1cm,innertopmargin=0.1cm,backgroundcolor=Apricot!60]\begin{defn1}[#1]%
        }{\end{defn1}\end{mdframed}%
        }
    \else
        \if 1\sidebar
        \tcolorboxenvironment{definition}{
        blanker,breakable,left=4mm,
        before skip=10pt,after skip=10pt,
        borderline west={2mm}{0pt}{Apricot}}
        \fi
    \fi
\fi

\if 1\specialremark
    \if 1\highlight
        \renewenvironment{remark}[1][]{%
        \begin{mdframed}[backgroundcolor=Salmon!60]\begin{rem1}[#1]%
        }{\end{rem1}\end{mdframed}%
        }
    \else
        \if 1\sidebar
        \tcolorboxenvironment{remark}{
        blanker,breakable,left=4mm,
        before skip=10pt,after skip=10pt,
        borderline west={2mm}{0pt}{Salmon}}
        \fi
    \fi
\fi

\if 1\colorequations
    \everymath{\color{MidnightBlue}}
\fi

\title{List Update with Delays or Time Windows}
 \author{Yossi Azar
     \thanks{School of Computer Science, Tel-Aviv University, Israel. Email: azar@tau.ac.il. Research supported in part by the Israel Science Foundation (grant No. 2304/20).}
 	\and
 	Shahar Lewkowicz
 	\thanks{School of Computer Science, Tel-Aviv University, Israel. Email: shaharlewko22@gmail.com.}
 	\and
 	Danny Vainstein%
 	\thanks{School of Computer Science, Tel-Aviv University, Israel and Google Research. Email: dannyvainstein@gmail.com.}
 }

\begin{document}

\maketitle

\begin{abstract}
We address the problem of \textbf{List Update}, which is considered one of the fundamental problems in online algorithms and competitive analysis. In this context, we are presented with a list of elements and receive requests for these elements over time. Our objective is to fulfill these requests, incurring a cost proportional to their position in the list. Additionally, we can swap any two consecutive elements at a cost of $1$. The renowned "Move to Front" algorithm, introduced by Sleator and Tarjan, immediately moves any requested element to the front of the list. They demonstrated that this algorithm achieves a competitive ratio of 2. While this bound is impressive, the actual cost of the algorithm's solution can be excessively high. For example, if we request the last half of the list, the resulting solution cost becomes quadratic in the list's length.

To address this issue, we consider a more generalized problem called \textbf{List Update with Time Windows}. In this variant, each request arrives with a specific deadline by which it must be served, rather than being served immediately. Moreover, we allow the algorithm to process multiple requests simultaneously, accessing the corresponding elements in a single pass. The cost incurred in this case is determined by the position of the furthest element accessed, leading to a significant reduction in the total solution cost. We introduce this problem to explore lower solution costs, but it necessitates the development of new algorithms. For instance, Move-to-Front fails when handling the simple scenario of requesting the last half of the list with overlapping time windows. In our work, we present a natural $O(1)$ competitive algorithm for this problem. While the algorithm itself is intuitive, its analysis is intricate, requiring the use of a novel potential function.

Additionally, we delve into a more general problem called \textbf{List Update with Delays}, where the fixed deadlines are replaced with arbitrary delay functions. In this case, the cost includes not only the access and swapping costs, but also penalties for the delays incurred until the requests are served. This problem encompasses a special case known as the prize collecting version, where a request may go unserved up to a given deadline, resulting in a specified penalty. For this more comprehensive problem, we establish an $O(1)$ competitive algorithm. However, the algorithm for the delay version is more complex, and its analysis involves significantly more intricate considerations.
\end{abstract}

\thispagestyle{empty}
\newpage
\setcounter{page}{1}
\section{Introduction}


One of the fundamental problems in online algorithms is the \textbf{List Update} problem. In this problem we are given an ordered list of elements and requests for these elements that arrive over time. Upon the arrival of a request, the algorithm must serve it immediately by accessing the required element. The cost of accessing an element is equal to its position in the list. Finally, any two consecutive elements in the list may be swapped at a cost of 1. The goal in this problem is to devise an algorithm so as to minimize the total cost of accesses and swaps. Note that it is an online algorithm and hence does not have any knowledge of future requests and must decide what elements to swap only based on requests that have already arrived.


Although the list update problem is a fundamental and simple problem, its solutions may be costly. Consider the following example. Assume that we are given requests to each of the elements in the farther half of the list. Serving these requests sequentially results in quadratic cost (quadratic in the length of the list). However, in many scenarios, while the requests arrive simultaneously, they do not have to be served immediately. Instead, they arrive with some deadline such that they must be served some time in the (maybe near) future. If this is the case, and the requests' deadlines are further in the future than their arrival, they may be jointly served; thereby incurring a linear (rather than quadratic) cost in the former example. This example motivates the following definition of List Update with Time Windows problem which may improve the algorithms' costs significantly.


The \textbf{List Update with Time Windows} problem is an extension of the classical List Update problem. Requests are once again defined as requests that arrive over time for elements in the list. However, in this problem they arrive with future deadlines. Requests must be served during their time window which is defined as the time between the corresponding request's arrival and deadline. This grants some flexibility, allowing an algorithm to serve multiple requests jointly at a point in time which lies in the intersection of their time windows. For a pictorial example, see Figure \ref{fig.timeWindowsExample}. The cost of serving a set of requests is defined as the current position of the farthest of those elements (i.e. serving a request for the $i$-th item in the list causes all the other active requests for the first $i$ elements in the list to be served as well in this access operation). In addition, as in the classical problem, swaps between any two consecutive elements may be performed at a cost of 1. Note that both accessing elements (or, serving requests) and swapping consecutive elements is done instantaneously (i.e., time does not advance during these actions). The goal is then to devise an online algorithm so as to minimize the total cost of serving requests and swapping elements. Also note that this problem encapsulates the original List Update problem. In particular, the List Update problem can be viewed as List Update with Time Windows where each time window consist of a unique single point.

We also consider a generalization of the time-window version - \textbf{List Update with Delays}. In this problem each request is associated with an arbitrary delay function, such that an algorithm accumulates delay cost while the request remains pending (i.e., unserved). The goal is to minimize the cost of serving the requests plus the total delay. This provides an incentive for the algorithm to serve the requests early. 

Another interesting and related variant is the prize collecting variant, which has been heavily researched in other fields as well. The price collecting problem is a special case of List Update with Delay and a generalization of List Update with Time Windows. In the context of List Update, the prize collecting problem is defined such that a request must be either served until some deadline or incur some penalty. Note that List Update with Delays encapsulates this variant by defining a delay function that incurs 0 cost and thereafter (at the deadline) immediately jumps to the penalty cost. The prize collecting problem encapsulates List Update with Time Windows when the penalty is arbitrarily large.
While the flexibility introduced in the list update with time windows or delays problems allow for lower cost solutions, it also introduces complexity in the considered algorithms. In particular, the added lenience will force us to compare different algorithms (our online algorithm compared to an optimal algorithm, for instance) at different time points in the input sequence. Since the problem definition allows for serving requests at different time points, this results in different sets of unserved requests when comparing the algorithms - this divergence will prove to be the crux of the problem and will result in significant added complexity compared to the classical List Update problem.

Originally, the List Update problem was defined to allow for free swaps to the accessed element: i.e., immediately after serving an element $e$, the algorithm may move $e$ towards the head of the list - free of charge. All other swaps between consecutive elements still incur a cost of 1. In our work, it will be convenient for us to consider the version of the problem where these free swaps are not allowed and all swaps between two consecutive elements incur a cost of 1. We would like to stress that while these two settings may seem different, this is not the case. One may easily observe that the difference in costs between a given solution in the two models is at most a multiplicative factor of 2. This can be seen to be true since the cost of the free swaps may be attributed to the cost of accessing the corresponding element (that was swapped) which is always at least as large. Thus, our results extend easily to the model with free swaps to the accessed element (by losing a factor of 2 in the competitive ratio). In particular, an algorithm which is constant competitive for one of the models is also constant-competitive for the other.

Using the standard definitions an \textbf{offline algorithm} sees the entire sequence of requests in advance and thus may leverage this knowledge for better solutions. Conversely, an \textbf{online algorithm} only sees a request (i.e., its corresponding element and entire time window or a delay function) upon its arrival and thus must make decisions based only on requests that have already arrived \footnote{In principle the time when a request arrives (i.e., is revealed to the online algorithm) need not be the same as the time when its time window or delay begins (i.e., when the algorithm may serve the request). Note however that the change makes no difference with respect to the offline algorithms but only allows for greater flexibility of the online algorithms. Therefore, any competitiveness results for our problem will transcend to instances with this change.}.
To analyze the performance of our algorithms we use the classical notion of \textbf{competitive ratio}. An online algorithm is said to be \(c\)-competitive (for \(c\geq 1\)) if for every input, the cost of the online algorithm is at most \(c\) times the cost of the optimal offline algorithm \footnote{We note that the lists of both the online algorithm and the optimum offline algorithm are identical at the beginning.}.

\subsection{Our Results}
In this paper, we show the following results:
\begin{itemize}
    \item For the List Update with Time Windows problem we provide a 24-competitive algorithm.
    \item For the List Update with Delays we provide a 336-competitive algorithm.   
\end{itemize}

For the \textbf{time windows} version the algorithm is natural. Upon a deadline of a request for an element, the algorithm serves all requests up to twice the element's position and then moves that element to the beginning of the list. Note that the algorithm does not use the fact that the deadline is known when the request arrives. I.e. our result holds even if the deadline is unknown until it is reached (as in non-clairvoyant models). Also note that while the algorithm is deceptively straightforward - its resulting analysis is tremendously more involved.

In the \textbf{delay} version the algorithm is more sophisticated. (See Appendix \ref{section_failed_algorithms} for counter examples to some simpler algorithms). The algorithm maintains two types of counters: request counters and element counters. For every request, its request counter increases over time at a rate proportional to the delay cost the request incurred. The request counter will be deleted at some point in time after the request has been served (it may not happen immediately after the request is served, but rather further in the future). Unlike the request counters, an element counter's scope is the entire time horizon. The element counter increases over time at a rate that is proportional to the sum of delay costs of unserved requests to that element. Once the requests are served, the element counter ceases to increase. There are two types of events that cause the algorithm to take action: prefix-request-counter events and element-counter events. A prefix-request-counter event takes place when the sum of the request counters of the first $\ell$ elements reaches a value of $\ell$. This event causes the algorithm to access the first $2\ell$ elements in the list and delete the request counters for requests to the first $\mathbf{\ell}$ elements. The request counters of the elements in positions $\ell+1$ up to $2\ell$ remain undeleted but cease to increase (Note that this will also result in the first $2\ell$ element counters to also cease to increase). An element-counter event takes place when an element counter's value reaches the element's position. Let $\ell$ be that position. This event causes the algorithm to access the first $2\ell$ elements in the list. Thereafter, the algorithm deletes all request counters of requests to that element. Finally, the element's counter is zeroed and the algorithm moves the element to the front.

It is interesting to note that List Update with Delay in the clairvoyant case can be reduced to the special case of prize collecting List Update (which is a generalization of List Update with Time Windows) by creating multiple requests with appropriate penalties. However, neither our algorithm for Delay nor our proof are getting simplified for this case, therefore we present our algorithm and proof for the general case (i.e. for List Update with Delay). Moreover, the reduction from List Update with Delay to prize collecting holds only for the clairvoyant case while our algorithm works on the non-clairvoyant model as well. In Appendix \ref{section_failed_algorithms} we give counter examples to the competitive ratio of simpler algorithms for List Update with Delay which hold also for the price collecting version.

\subsection{Previous Work}
We begin by reviewing previous work relating to the classical List Update problem. Sleator and Tarjan \cite{sleator1985amortized} began this line of work by introducing the deterministic online algorithm Move to Front (i.e. \(MTF\)). Upon a request for an element \(e\), this algorithm accesses \(e\) and then moves \(e\) to the beginning of the list. They proved that \(MTF\) is \(2\)-competitive in a model where free swaps to the accessed element are allowed. The proof uses a potential function defined as the number of inversions between \(MTF\)'s list and \(OPT\)'s list. An inversion between two lists is two elements such that their order in the first list is opposite to their order in the second list. A simple lower bound of \(2\) for the competitive ratio of deterministic online algorithms is achieved when the adversary always requests the last element in the online algorithm's list and \(OPT\) orders the elements in its list according to the number of times they were requested in the sequence. Since the model with no free swaps differs in the cost by at most a factor of \(2\) this immediately yields that \(MTF\) is \(4\)-competitive for the model with no free swaps. The simple \(2\) lower bound also holds for this model. Previous work regarding randomized upper bounds for the competitive ratio have been done by many others \cite{irani1991two, reingold1994randomized, albers1997revisiting, albers1998improved, ambuhl2000optimal}. Currently, the best known competitiveness was given by Albers, Von Stengel, and Werchner \cite{albers1995combined}, who presented a random online algorithm and proved it is \(1.6\) competitive. Previous work regarding lower bounds for this problem have also been made \cite{teia1993lower, reingold1994randomized, ambuhl2000optimal}. The highest of which was achieved by Amb{\"u}hl, Gartner and Von Stengel \cite{ambuhl2001new}, who proved a lower bound of \(1.50084\) on the competitive ratio for the classical problem. 
With regards to the offline classical problem: Amb{\"u}hl proved this problem is NP-hard \cite{ambuhl2000offline}. 

Problems with time windows have been considered for various online problems. Gupta, Kumar and Panigrahi \cite{Caching_with_Time_Windows_and_Delays} considered the problem of paging (caching) with time windows. Bienkowski et al. \cite{Online_Algorithms_for_Multi_Level_Aggregation} considered the problem of online multilevel aggregation. Here, the problem is defined via a weighted rooted tree. Requests arrive on the tree's leaves with corresponding time windows. The requests must be served during their time window. Finally, the cost of serving a set of requests is defined as the weight of the subtree spanning the nodes that contain the requests. Bienkowski et al. \cite{Online_Algorithms_for_Multi_Level_Aggregation} showed a $O(D^4 2^D)$ competitive algorithm where $D$ denotes the depth of the tree. Buchbinder et al. \cite{CompetitiveAlgorithmforOnlineMultilevelAggregation} improved this to $O(D)$ competitiveness. Later, Azar and Touitou \cite{General_Framework_for_Metric_Optimization_Problems_with_Delay_or_with_Deadlines, focs_AzarT20} provided a framework for designing and analyzing algorithms for these types of metric optimization problems. 

In addition, set cover with deadline \cite{SetCoverwithDelayClairvoyanceIsNotRequired} was also considered as well as online service in a metric space \cite{BienkowskiKS2018, AzarGGP2017}. To all these problems poly-logarithmic competitive algorithms were designed. It is interesting to note that in contrast to all these problems we show that for our list update problem constant competitive algorithms are achievable.
We note that problems with deadline can be also extended to problems with delay where there is a monotone penalty function for each request that is increasing over time until the request is served (and is added to the original cost). Many of the results mentioned above can be extended to arbitrary penalty function.
The main exception is matching with delays that can be efficiently solved (i.e. with poly-logarithmic competitive ratio) only for linear functions \cite{Online_matching_haste_makes_waste, Polylogarithmic_Bounds_on_the_Competitiveness_of_Min_cost_Perfect_Matching_with_Delays, AshlagiACCGKMWW2017} as well as for concave  functions \cite{TheMinCostMatchingwithConcaveDelaysProblem}. 
For other problems that tackle deadlines and delays see: 
\cite{OnlineFacilityLocationwithLinearDelay, On_bin_packing_with_clustering_and_bin_packing_with_delays, The_Price_of_Clustering_in_Bin-Packing_with_Applications_to_Bin_Packing_with_Delays,ApproximationAlgorithmsfortheJointReplenishmentProblemwithDeadlines,BetterApproximationBoundsfortheJointReplenishmentProblem}.

\subsection{Our Techniques}
\label{section_techniques}

While introducing delays or time windows introduces the option of serving multiple requests simultaneously thereby drastically improving the solution costs, this lenience requires the algorithms and their analyses to be much more intricate. 

The "freedom" given to the algorithm compared with the classical List Update problem requires more decisions to be made: for example, in the time windows version assume there are currently two active requests: a request for an element \(e_1\) which just reached its deadline and a request for a further element in the list, \(e_2\) but its deadline has not been reached yet. Should the algorithm access only \(e_1\), pay its position in the list and leave the request for \(e_2\) to be served later or access both \(e_1\) and \(e_2\) together and pay the position of \(e_2\) in the list?
If no more requests arrive until the deadline of the second active request, the latter option is better. However, requests that might arrive before the deadline of the second active request might cause the former option to be better after all.
In the delay version the decision is more complicated since it may be the case that there are various requests for elements, each request accumulated a small or medium delay but their total is large. Hence, we need to decide at what stage and to what extend serving these requests. Moreover it is more tricky to decide which element to move to the front of the list and at which point in time.

As for the analysis, we need to handle the fact that the online algorithm and the optimal algorithm serve requests at different times. Further, since both algorithms may serve different sets of requests at different times, we may encounter situations wherein a given request at a given time would have been served by the online algorithm and not the optimal algorithm (and vice versa). This, combined with the fact that the algorithms’ lists may be ordered differently at any given time, will prove to be the crux of our problem and its analysis.

To overcome these problems, we introduce new potential functions (one for the time windows case and one for the delays case). We note that the original List Update problem was also solved using a potential function \cite{sleator1985amortized}, however, due to the aforementioned issues, the original function failed to capture the resulting intricacies and we had to introduce novel (and more involved) functions. Ultimately, this resulted in constant competitiveness for both settings.

\textbf{List Update with Time Windows}: Here, the potential function consists of three terms. The first accounts for the difference (i.e., number of inversions) between the online and optimal algorithms’ lists at any given time (similar to that of Sleator and Tarjan \cite{sleator1985amortized}). The second term accounts for the difference in the set of served requests between the two algorithms. Specifically, whenever the optimal algorithm serves a request not yet served by the online algorithm, we add value to this term which will be subtracted once the online algorithm serves the request. The third term accounts for the movement costs made by the online algorithm incurred by requests that were already served by the optimal algorithm.

At any given time point, our proof considers separately elements that are positioned (significantly) further in the list in the online algorithm compared to the optimal algorithm, as opposed to all other elements (which we will refer to as “the closer” elements). To understand the flavor of our proofs, e.g., the incurred costs of “the further” elements is charged to the first term of the potential function. In contrast, the change in the first term is not be enough to cover the incurred costs of “the closer” elements (the term may even increase). Fortunately, the second term is indeed enough to cover both the incurred costs and the (possible) increase in the first term. Specifically, the added value is of the same order of magnitude as the access cost incurred by the optimal algorithm for serving the corresponding requests. This follows from (a) only requests for elements in $ALG$ which are located at a position which is of the same order of magnitude as the location in $OPT$ get "gifts" in the second term. (b) The fact that the number of trigger elements and their positions in $ALG$s list is bounded because upon a deadline of a trigger, $ALG$ serves all the elements located up to twice the position of the trigger in its list (The definition of the term "trigger" appears in the beginning of Section \ref{section_algorithm_deadlines}).

Note however that the analysis above holds only as long as the optimal algorithm does not move an element further in the list between the time it serves it and the time the online algorithm serves it. In such a case, the third term will offset the costs.

\textbf{List Update with Delays}: Here, the potential function consists of five terms. The first term is similar to that of the time windows setting with the caveat that defining the distance between the online and optimal algorithms’ lists should depend on the values of the element counters as well. 
Consider the following example. Assume that the ordering of $i,j$ is reversed when comparing it between the online and optimal algorithms and assume it is ordered $(i,j)$ in the online algorithm. As we defined our algorithm, once the element counter of $j$ is filled, it is moved to the front and therefore the ordering will be reversed. Therefore, intuitively, if $j$ element counter is almost filled we consider the distance between this pair smaller than the case where its element counter is completely empty. Therefore, we would like the contribution to the potential function to be smaller in the former case.

Note that the contribution of the inversion $(i,j)$ depends on the element counter of $j$ but not on the element counter of $i$ (i.e. the contribution is asymmetric). Even if the element counter of $j$ is very close to its position in the online algorithm’s list, we still need a big contribution of the pair $(i,j)$ in order to pay for the next element counter event on $j$. However, if the element counter of $j$ is far from its position in the online algorithm’s list, we will need even more contribution of the pair $(i,j)$ to the potential function in order to also cover future delay penalty which the algorithm may suffer on the element $j$ that will not cause an element counter event on $j$ to occur in the short term.

The second part of the potential function consists of the delay cost that both the online and optimal algorithms incurred for requests which were active in both algorithms. This term is used to cover the next element counter events for the elements required in these requests. The third part of the potential function offsets the requests which have been served by the optimal algorithm but not by the online algorithm. This part is very similar to the gifts in the second term of the potential function in time windows and the ideas behind it are similar. Again, the gifts are only given to requests which are located by the online algorithm at a position which is of the same order of magnitude as the location in the optimal algorithm. The gift is of the same order of magnitude as the total delay the online algorithm pays for the request (including the delay it will pay in the future). This is used in order to offset the next element counter event in the online algorithm on the element. However, this gift also decreases as the online algorithm suffers more delay for the request because we want this term in the potential function to also cover the future delay penalty the online algorithm will pay for the request.

The fourth and fifth terms in the potential function are very similar to the third term in the potential function of time windows but each one of them has its own purposes: The fourth term should cover the next element counter event on the element while the fifth term should cover the scenario in which the optimal algorithm served a request and then moved the element further in its list but the online algorithm will suffer more delay penalty for this request in the future. The fifth term should cover this delay cost that the online algorithm pays and thus it is proportional to the fraction between the future delay the online algorithm pays for the request and the position of the element in the online algorithm’s list.

\section{The Model for Time Windows and Delays}

Given an input $\sigma$ and algorithm $ALG$ we denote by $ALG(\sigma)$ the cost of its solution. Recall that in the \textbf{time windows} setting $ALG(\sigma)$ is defined as the sum of (1) the algorithm's access cost: the algorithm may serve multiple requests at a single time point and then the access cost is defined as the position of the farthest element in this set of requests.
$ALG(\sigma)$ also accounts for (2) the total number of element swaps performed by $ALG$. In total, $ALG(\sigma)$ is equal to the sum of access costs and swaps. In the \textbf{delay} setting $ALG(\sigma)$ accounts (1) and (2) as above in addition to (3) the sum of the delay incurred by all requests. The delay is defined via a delay function that is associated with each request. The delay functions may be different per request and are each a monotone non-decreasing non-negative function. In total, $ALG(\sigma)$ is equal to the sum of access costs, swaps and delay costs. As is traditional when analysing online algorithms, we denote by $OPT(\sigma)$ the cost of the optimal solution to input $\sigma$. Furthermore, we say that $ALG$ is $c$-competitive (for $c \geq 1$) if for every input $\sigma$, $ALG(\sigma) \leq c \cdot OPT(\sigma)$. Throughout our work, when clear from context, we use $ALG(\sigma)$ to denote both the cost of the solution and the solution itself. Our algorithms work also in the non-clairvoyant case: In the time windows version we only know the deadline of a request upon its deadline (and not upon its arrival). In the delay version we know the various delay functions of the requests only up to the current time. Next we introduce several notations that will aid us in our proofs.




\begin{definition}
Let $\mathbb{E}$ be the set of the elements.
\begin{itemize}
    \item Let ${n}$ denote the number of elements in our list ($|\mathbb{E}|=n$) and ${m}$ the number of requests.
    \item Let ${r_k}$ denote the $k$th request and ${e_k}$ the requested element by $r_k$.
    \item Let $y_k\in[n]$ denote the position of $e_k$ in $OPT$s list at the time $OPT$ served $r_k$. Let $x_k\in[n]$ denote the position of $e_k$ in $ALG$s list at the time $OPT$ (and not $ALG$) served $r_k$\footnote{In the delay version, $x_k$ and $y_k$ are defined only in case $OPT$ indeed served the request $r_k$ at some time.}.
\end{itemize} 
\end{definition}
Throughout our work, given an element in the list, we oftentimes consider its neighboring elements in the list. We therefore introduce the following conventions to avoid confusion. Given an element in the list we refer to its \textbf{previous} element as its neighbor which is closer to the head of the list and its \textbf{next} element as its neighbor which is further from the head of the list.
\section{The Algorithm for Time Windows}
\label{section_algorithm_deadlines}

\noindent Prior to defining our algorithm, we need the following definitions.

\begin{definition}
We define the \textbf{triggering element}, when a deadline of a request is reached, as the farthest element in the list such that there exists an active request for it which just reached its deadline. We define the \textbf{triggering request} as one of the active requests for the triggering element that just reached its deadline - arbitrary.
\end{definition}

When clear from context we will use the term "trigger" instead of "triggering request" or "triggering element". Next, we define the algorithm.

\begin{algorithm}[H]
\caption{Algorithm for Time Windows (i.e. Deadlines)}
\label{deadlines.alg}

\textbf{Upon} deadline of a request \textbf{do}:

\Indp $i \leftarrow $ triggering element's position

Serve the set of requests in the first $2i-1$ elements in the list

Move-to-front the triggering element

\Indm







\end{algorithm}

We prove the following theorem for the above algorithm in Appendix \ref{section_deadlines}.
\begin{theorem}
\label{deadlines.thm.1}
For each sequence of requests $\sigma$, we have that $$ALG(\sigma)\leq 24\cdot OPT(\sigma).$$
\end{theorem}
\section{The Algorithm for Delays}
\label{section_algorithm_delay}

Our algorithm maintains two types of counters in order to process the input: requests counters and element counters. We begin by defining the \textbf{request counters}. The algorithm  maintains a separate request counter for every incoming request. For a given request $r_k$ we denote its corresponding counter as $RC_k$. The counter is initialized to 0 the moment the request arrives and increases at the same rate that the request incurs delay. Once the request is served, the counter ceases to increase. Finally, our algorithm deletes the request counters - it will do so at some point in the future after the request is served (but not necessarily immediately when the request is served).

Next we define the \textbf{element counters}. Unlike the request counters, element counters exist throughout the entire input (i.e., they are initialized at the start of the input and do not get deleted). We define an element counter ${EC}_e$ for every element $e\in\mathbb{E}$. These counters are initialized to 0 and increase at a rate equal to the total delay incurred by requests to the specific element.

We define two types of events that cause the algorithm to act: prefix-request-counter events and element-counter events. A \textbf{prefix request counters event on $\ell$} for $\ell \in [n]$ occurs when the sum of all the request counters of requests for the first $\ell$ elements in the list reaches the value of $\ell$. When this type of event takes place, the algorithm performs the following two actions. First, it serves the requests of the first $2\ell$ elements. Second, it deletes the request counters that belong to the first $\ell$ elements. Note that these are the request elements that contributed to this event and are therefore deleted. Also note that the request counters of the elements $\ell + 1$ to $2\ell$ and the element counters of the first $2\ell$ elements cease to increase since their requests have been served.

An \textbf{element counter event on $e$} for $e\in\mathbb{E}$ occurs when ${EC}_e$ reaches the value of $\ell$, where $\ell\in[n]$ is the position of the element $e$ in the list, currently. When this type of event takes place, the algorithm performs the following three actions. First, it serves the requests on the first $2\ell$ elements. Second, it deletes all request counters of requests to the element $e$. Third, it sets ${EC}_e$ to 0 and perform move-to-front to $e$.

Note that the increase in an element counter equals to the sum of the increase of all the request counters to this element. In particular, the value of the element counter is at least the sum of the non-deleted request counters for the element (It may be larger since request counters may be deleted in request counters events while the element counter maintains its value). Hence when we zero an element counter, we also delete the request counters of requests for this element in order to maintain this invariant.

Next, we present the algorithm.

\begin{algorithm}[H]
\caption{Algorithm for Delay}
\label{delay.alg}

\textbf{Initialization:}

\Indp

\textbf{For each} $e\in\mathbb{E}$ \textbf{do}:

\Indp

${EC}_e\leftarrow 0$

\Indm

\Indm

\textbf{Upon} arrival of a new request $r_k$ \textbf{do}:

\Indp 

${RC}_k\leftarrow 0$

\Indm

\textbf{Upon} prefix-request-counters event on $\ell\in[n]$ \textbf{do}:

\Indp

Serve the set of requests in the first $2\ell$ elements in the list

Delete the request counters for the first $\ell$ elements in the list

\Indm

\textbf{Upon} element-counter event on $e$ (let $\ell$ denote $e$'s current position) \textbf{do}:

\Indp

Serve the set of requests in the first $2\ell$ elements in the list

Delete all the request counters of requests for the element $e$

${EC}_e\leftarrow 0$

Move-to-front the element $e$

\Indm
\end{algorithm}

\noindent We prove the following theorem for the above algorithm in Appendix \ref{section_delay}.
\begin{theorem}
\label{delay.thm.1}
For each sequence of requests $\sigma$, we have that $$ALG(\sigma)\leq 336\cdot OPT(\sigma).$$
\end{theorem}

\section{Potential Functions for Time Windows and Delay}

Our proofs use potential functions. In particular we prove for each possible event that $$\Delta ALG+\Delta \Phi\leq c\cdot\Delta OPT$$ where $\Phi$ is the potential and $c$ is the competitive ratio. In this section we describe the potential functions. The detailed proofs that use these potential functions appear in Appendix \ref{section_deadlines} and Appendix \ref{section_delay}.

\subsection{Time Windows}
\label{section_potential_function_deadlines}
As mentioned earlier, our potential function used for the time windows setting is comprised of three terms. We will define them separately. We begin with the first term that aims to capture the difference between $ALG$ and $OPT$'s lists at any given moment.

\begin{definition}
    Let ${\phi(t)}$ denote the number of \textbf{inversions} between $ALG$'s and $OPT$'s lists at time $t$. Specifically, $\phi(t)=|\{(i,j)\in \mathbb{E}^2|\text{ At time }t \text{, }i \text{ is before } j \text{ in } ALG \text{'s list and after } j \text{ in } OPT \text{'s list}\}|$.
\end{definition}

The second term accounts for the difference in the set of served requests between the two algorithms. Specifically, whenever the optimal algorithm serves a request not yet served by the online algorithm, we add value to this term which will be subtracted
once the online algorithm serves the request. Before defining this term, we need the following definition.

\begin{definition}
For each time \(t\), let
\(\lambda(t)\subseteq[m]\) be the set of all the request indices \(k\) such that the request \(r_k\) arrived and was served by \(OPT\) but was not served by \(ALG\) at time \(t\).
\end{definition}

Recall that for request $r_k$ we denote by $y_k$ the position of $e_k$ in $OPT$'s list at the time that $OPT$ served $r_k$. Furthermore, we denote by $x_k$ the position of $e_k$ in $ALG$'s list at the time $OPT$ (and not $ALG$) served $r_k$. We are now ready to define the second term in our potential function.

\begin{definition}
    For $k \in \lambda (t)$ we define \(\psi(x_k, y_k)\geq 0\) as 
    \[
    \psi(x,y) = \begin{cases}
    7x & \text{if  } 1\leq x\leq y \\
    8y-x & \text{if  } y\leq x\leq 8y \\
    0 & \text{if  } 8y\leq x
    \end{cases}
    \]
\end{definition}

\noindent Next, we define the third term of our potential function.

\begin{definition}
    We define \(\mu_k(t)\) as the number of swaps \(OPT\) performed between \(e_k\) and its next element in the list from the time \(OPT\) served the request \(r_k\) until time \(t\). 
\end{definition}

\noindent Finally, we combine the terms and define our potential function.

\begin{definition}
We define our potential function for Time Windows as
\[\Phi(t)=4\cdot \phi(t)+\sum_{k\in\lambda(t)}\psi(x_k,y_k)+4\cdot\sum_{k\in\lambda(t)}\mu_k(t).\]
\end{definition}

\subsection{Delay}
\label{section_potential_function_delay}

In the delays setting, we define a different potential function that is comprised of five terms. We will define the terms separately first and thereafter use them to compose our potential function. We begin with the first term.

As mentioned in Our Techniques, the first term also aim to capture the distance between $ALG$'s and $OPT$'s lists. In the time windows setting, we defined this term as the number of element inversions. In the delays case this does not suffice; we have to take into the account the elements' counters as well. To gain some intuition as to why this addition is needed, consider the following example. Assume that elements $i,j$ are ordered $(i,j)$ in $ALG$ and reversed in $OPT$. Recall that $ALG$ is defined such that when $j$' element counter is filled, then we move it to the front (thereby changing the $ALG$'s ordering to $(j,i)$). Therefore, if it is the case that $j$'s element counter is nearly filled, intuitively we may say that $i,j$'s ordering in $ALG$ and $OPT$ are closer to each other than if $j$'s element counter would have been empty. Therefore, we would like the contribution to the potential function to be smaller in the former case.

Note that the contribution of the inversion $(i,j)$ depends on the element counter of $j$ but not on the element counter of $i$ (i.e. the contribution is asymmetric). Even if the element counter of $j$ is very close to its position in the online algorithm’s list, we still need a big contribution of the pair $(i,j)$ in order to pay for the next element counter event on $j$. However, if the element counter of $j$ is far from its position in the online algorithm’s list, we need even more contribution of the pair $(i,j)$ to the potential function in order to also cover future delay penalty which the algorithm may suffer on the element $j$ that does not cause an element counter event on $j$ to occur in the short term. Before formally defining this term, we define the following.

\begin{definition}
For a time $t$ and an element $e\in\mathbb{E}$ we define:
\begin{itemize}
    \item $EC_e^t$ to be the value of the element counter $EC_e$ at time $t$.
    \item $x_e^t\in[n]$ ($y_e^t\in[n]$ resp) to be the position of $e$ in $ALG$s ($OPT$s resp) list at time $t$.
    \item $I_e^t=\{i\in\mathbb{E}|i \text{ is before } e \text{ in } ALG\text{s list and after } e \text{ in } OPT\text{s list at time } t\}$.
\end{itemize}
\end{definition}

\begin{definition}
    For element $e \in \mathbb{E}$ we define $\rho_e(t) =|I_e^t|\cdot(28-8\cdot\frac{EC_e^t}{x_e^t})$
\end{definition}

Observe that each $i\in I_e^t$ contributes $20+8\cdot(1-\frac{EC_e^t}{x_e^t})$ to $\rho_e(t)$. The additive term of $20$ is used in order to cover the next element counter event for $e$ while the second term is used to cover the delay penalty $ALG$ will pay in the future for requests for $e$. Note that the term $1-\frac{EC_e^t}{x_e^t}$ is the fraction of $EC_e$ which is not "filled" yet. If this term is very low, $ALG$ is very close to have an element counter event on $e$, which causes the order of $i$ and $e$ in $ALG$s list and $OPT$s list to be the same, thus it makes sense that the contribution of $i$ to $\rho_e(t)$ is lower compared with the case where $1-\frac{EC_e^t}{x_e^t}$ would be higher.

Next, we consider the second term. First, we denote the total incurred delay by a request as $d_k(t)$. Formally, this is defined as follows.

\begin{definition}
For a given request $r_k$ and time $t$ let $d_k(t)$ denote the total delay incurred by the request by $ALG$ up to time $t$. (Note that it is defined as 0 before the request arrived and remains unchanged after the request is served). Let $d_k=\sup_{t}d_k(t)$. Note that this is a supremum and not maximum for the case that $r_k$ is never served. Note that $d_k\leq n$ because $ALG$ always serves $r_k$ before $d_k>n$.
\end{definition}

Our second term is a sum of incurred delay costs of specific elements. 

\begin{definition}
For each $k\in[m]$, the request $r_k$ is considered:
\begin{itemize}
    \item \textbf{active} in $ALG$ (resp. $OPT$) from the time it arrives until it is served by $ALG$ (resp. $OPT$).
    \item \textbf{frozen} from the time it is served by $ALG$ until ${EC}_{e_k}$ is zeroed in an $e_k$ element counter event.
\end{itemize}
\end{definition}

\begin{definition}
For time $t$ we define $\lambda(t)\subseteq[m]$ as the set of requests (request indices) which are either active \textbf{or} frozen in $ALG$ at time $t$. We define $\lambda_1(t) \subseteq \lambda(t)$ as the set of requests that are also active in $OPT$ at time $t$ and 
$\lambda_2(t) \subseteq \lambda(t)$ as the set of requests that are also not active in $OPT$ at time $t$.
\end{definition}

\noindent Finally, we define our second term.

\begin{definition}
    We define the second term of the Delays potential function as $\sum_{k\in\lambda_1(t)}d_k(t)$.
\end{definition}

\noindent The third term is defined as follows (we use $x_k$ and $y_k$ as previously defined).

\begin{definition}
    We define the third term as $\sum_{k\in\lambda_2(t)}(42d_k-6d_k(t))\cdot\mathbbm{1}{[x_k\leq4y_k]}$.
\end{definition}

Note that $42d_k-6d_k(t)=36d_k+6\cdot(d_k-d_k(t))$. Therefore each request index $k\in\lambda_2(t)$ contributes two terms to $\Phi$: $36d_k$ is used to cover the next element counter on $e_k$ while the second term is $6$ times the delay $ALG$ will pay for $r_k$ in the future, which will be used to cover this exact delay penalty that $ALG$ will pay in the future for $r_k$.

The fourth term is defined to cover the next element counter event on a given element as follows.

\begin{definition}
    Let $\mu_e(t)$, for $e \in \mathbb{E}$, be the number of swaps $OPT$ performed between $e$ and its next element in its list ever since the last element counter event before time $t$ on $e$ by $ALG$ (or the beginning of the time horizon if there was not such an event).
\end{definition}

Finally, we define the fifth term. The fifth term should cover the scenario in which the optimal algorithm served a request and then moved the element further in its list but the online algorithm will suffer more delay penalty for this request in the future. It will also cover the delay cost that the online algorithm will pay and thus it is proportional to the fraction between the future delay the online algorithm will pay for the request and the position of the element in the online algorithm’s list.
\begin{definition}
    Let $\mu_k(t)$, for $k\in\lambda_2(t)$, be the number of swaps $OPT$ performed between $e_k$ and its next element in its list ever since $OPT$ served the request $r_k$ (by accessing $e_k$).
\end{definition}
\begin{definition}
    We define the fifth term of the Delays potential function as $8\cdot\sum_{k\in\lambda_2(t)}\frac{d_k-d_k(t)}{x_{e_k}^t}\cdot\mu_k(t)$.
\end{definition}

\noindent We are now ready to define our potential function.

\begin{definition}
We define our potential function for the delays setting as 
\begin{align*}
    \Phi(t) & = \sum_{e\in\mathbb{E}}\rho_e(t)+36\cdot\sum_{k\in\lambda_1(t)}d_k(t)+\sum_{k\in\lambda_2(t)}(42d_k-6d_k(t))\cdot\mathbbm{1}{[x_k\leq4y_k]}+\\
    & +48\cdot\sum_{e\in\mathbb{E}}\mu_e(t)+8\cdot\sum_{k\in\lambda_2(t)}\frac{d_k-d_k(t)}{x_{e_k}^t}\cdot\mu_k(t)
\end{align*}
\end{definition}
\section{Conclusion and Open Problems}
In this paper, we presented the List Update with Time Windows and Delay, which generalize the classical List Update problem.
\begin{itemize}
\item 
We presented a $24$-competitive ratio algorithm for the List Update with Time Windows problem. 
\item 
We presented a $336$-competitive ratio algorithm for the List Update with Delays problem. 
\item Open problems: The main issue left unsolved is the gap between the upper and lower bounds. Currently, the best lower bound for both problems considered is 2. Note that this is the same lower bound given to the original List Update problem. An interesting followup would be to improve upon this result and show a better lower bound. On the other hand, one may improve the upper bound - our algorithms are non-clairvoyant in the sense that our proofs and algorithms hold even when the deadlines/delays are unknown. It would be interesting to understand whether clairvoyance may improve the upper bound. Another interesting direction would be to consider randomization as a way of improving our bounds.

\end{itemize}
\bibliographystyle{plainnat}
\bibliography{bib.bib}

\begin{thebibliography}{27}
\providecommand{\natexlab}[1]{#1}
\providecommand{\url}[1]{\texttt{#1}}
\expandafter\ifx\csname urlstyle\endcsname\relax
  \providecommand{\doi}[1]{doi: #1}\else
  \providecommand{\doi}{doi: \begingroup \urlstyle{rm}\Url}\fi

\bibitem[Albers(1998)]{albers1998improved}
Susanne Albers.
\newblock Improved randomized on-line algorithms for the list update problem.
\newblock \emph{SIAM Journal on Computing}, 27\penalty0 (3):\penalty0 682--693,
  1998.

\bibitem[Albers and Mitzenmacher(1997)]{albers1997revisiting}
Susanne Albers and Michael Mitzenmacher.
\newblock Revisiting the counter algorithms for list update.
\newblock \emph{Information processing letters}, 64\penalty0 (3):\penalty0
  155--160, 1997.

\bibitem[Albers et~al.(1995)Albers, Von~Stengel, and
  Werchner]{albers1995combined}
Susanne Albers, Bernhard Von~Stengel, and Ralph Werchner.
\newblock A combined bit and timestamp algorithm for the list update problem.
\newblock \emph{Information Processing Letters}, 56\penalty0 (3):\penalty0
  135--139, 1995.

\bibitem[Amb{\"u}hl(2000)]{ambuhl2000offline}
Christoph Amb{\"u}hl.
\newblock Offline list update is np-hard.
\newblock In \emph{European Symposium on Algorithms}, pages 42--51. Springer,
  2000.

\bibitem[Amb{\"u}hl et~al.(2000)Amb{\"u}hl, G{\"a}rtner, and
  Stengel]{ambuhl2000optimal}
Christoph Amb{\"u}hl, Bernd G{\"a}rtner, and Bernhard~von Stengel.
\newblock Optimal projective algorithms for the list update problem.
\newblock In \emph{International Colloquium on Automata, Languages, and
  Programming}, pages 305--316. Springer, 2000.

\bibitem[Amb{\"u}hl et~al.(2001)Amb{\"u}hl, G{\"a}rtner, and
  Von~Stengel]{ambuhl2001new}
Christoph Amb{\"u}hl, Bernd G{\"a}rtner, and Bernhard Von~Stengel.
\newblock A new lower bound for the list update problem in the partial cost
  model.
\newblock \emph{Theoretical Computer Science}, 268\penalty0 (1):\penalty0
  3--16, 2001.

\bibitem[Ashlagi et~al.(2017)Ashlagi, Azar, Charikar, Chiplunkar, Geri, Kaplan,
  Makhijani, Wang, and Wattenhofer]{AshlagiACCGKMWW2017}
Itai Ashlagi, Yossi Azar, Moses Charikar, Ashish Chiplunkar, Ofir Geri, Haim
  Kaplan, Rahul~M. Makhijani, Yuyi Wang, and Roger Wattenhofer.
\newblock Min-cost bipartite perfect matching with delays.
\newblock In \emph{{APPROX/RANDOM}}, pages 1:1--1:20, 2017.

\bibitem[Azar and
  Touitou(2019)]{General_Framework_for_Metric_Optimization_Problems_with_Delay_or_with_Deadlines}
Yossi Azar and Noam Touitou.
\newblock General framework for metric optimization problems with delay or with
  deadlines.
\newblock In David Zuckerman, editor, \emph{60th {IEEE} Annual Symposium on
  Foundations of Computer Science, {FOCS} 2019, Baltimore, Maryland, USA,
  November 9-12, 2019}, pages 60--71. {IEEE} Computer Society, 2019.
\newblock \doi{10.1109/FOCS.2019.00013}.
\newblock URL \url{https://doi.org/10.1109/FOCS.2019.00013}.

\bibitem[Azar and Touitou(2020)]{focs_AzarT20}
Yossi Azar and Noam Touitou.
\newblock Beyond tree embeddings - a deterministic framework for network design
  with deadlines or delay.
\newblock In Sandy Irani, editor, \emph{61st {IEEE} Annual Symposium on
  Foundations of Computer Science, {FOCS} 2020, Durham, NC, USA, November
  16-19, 2020}, pages 1368--1379. {IEEE}, 2020.
\newblock \doi{10.1109/FOCS46700.2020.00129}.
\newblock URL \url{https://doi.org/10.1109/FOCS46700.2020.00129}.

\bibitem[Azar et~al.(2017{\natexlab{a}})Azar, Chiplunkar, and
  Kaplan]{Polylogarithmic_Bounds_on_the_Competitiveness_of_Min_cost_Perfect_Matching_with_Delays}
Yossi Azar, Ashish Chiplunkar, and Haim Kaplan.
\newblock Polylogarithmic bounds on the competitiveness of min-cost perfect
  matching with delays.
\newblock In \emph{Proceedings of the Twenty-Eighth Annual {ACM-SIAM} Symposium
  on Discrete Algorithms, {SODA} 2017, Barcelona, Spain, Hotel Porta Fira,
  January 16-19}, pages 1051--1061, 2017{\natexlab{a}}.

\bibitem[Azar et~al.(2017{\natexlab{b}})Azar, Ganesh, Ge, and
  Panigrahi]{AzarGGP2017}
Yossi Azar, Arun Ganesh, Rong Ge, and Debmalya Panigrahi.
\newblock Online service with delay.
\newblock In \emph{{STOC}}, pages 551--563, 2017{\natexlab{b}}.

\bibitem[Azar et~al.(2019)Azar, Emek, van Stee, and
  Vainstein]{The_Price_of_Clustering_in_Bin-Packing_with_Applications_to_Bin_Packing_with_Delays}
Yossi Azar, Yuval Emek, Rob van Stee, and Danny Vainstein.
\newblock The price of clustering in bin-packing with applications to
  bin-packingwith delays.
\newblock In \emph{The 31st {ACM} on Symposium on Parallelism in Algorithms and
  Architectures, {SPAA} 2019, Phoenix, AZ, USA, June 22-24, 2019}, pages 1--10,
  2019.

\bibitem[Azar et~al.(2020)Azar, Chiplunkar, Kutten, and
  Touitou]{SetCoverwithDelayClairvoyanceIsNotRequired}
Yossi Azar, Ashish Chiplunkar, Shay Kutten, and Noam Touitou.
\newblock Set cover with delay - clairvoyance is not required.
\newblock In Fabrizio Grandoni, Grzegorz Herman, and Peter Sanders, editors,
  \emph{28th Annual European Symposium on Algorithms, {ESA} 2020, September
  7-9, 2020, Pisa, Italy (Virtual Conference)}, volume 173 of \emph{LIPIcs},
  pages 8:1--8:21. Schloss Dagstuhl - Leibniz-Zentrum f{\"{u}}r Informatik,
  2020.
\newblock \doi{10.4230/LIPIcs.ESA.2020.8}.
\newblock URL \url{https://doi.org/10.4230/LIPIcs.ESA.2020.8}.

\bibitem[Azar et~al.(2021)Azar, Ren, and
  Vainstein]{TheMinCostMatchingwithConcaveDelaysProblem}
Yossi Azar, Runtian Ren, and Danny Vainstein.
\newblock The min-cost matching with concave delays problem.
\newblock In D{\'{a}}niel Marx, editor, \emph{Proceedings of the 2021
  {ACM-SIAM} Symposium on Discrete Algorithms, {SODA} 2021, Virtual Conference,
  January 10 - 13, 2021}, pages 301--320. {SIAM}, 2021.
\newblock \doi{10.1137/1.9781611976465.20}.
\newblock URL \url{https://doi.org/10.1137/1.9781611976465.20}.

\bibitem[Bienkowski et~al.(2013)Bienkowski, Byrka, Chrobak, Dobbs, Nowicki,
  Sviridenko, Swirszcz, and
  Young]{ApproximationAlgorithmsfortheJointReplenishmentProblemwithDeadlines}
Marcin Bienkowski, Jaroslaw Byrka, Marek Chrobak, Neil~B. Dobbs, Tomasz
  Nowicki, Maxim Sviridenko, Grzegorz Swirszcz, and Neal~E. Young.
\newblock Approximation algorithms for the joint replenishment problem with
  deadlines.
\newblock In Fedor~V. Fomin, Rusins Freivalds, Marta~Z. Kwiatkowska, and David
  Peleg, editors, \emph{Automata, Languages, and Programming - 40th
  International Colloquium, {ICALP} 2013, Riga, Latvia, July 8-12, 2013,
  Proceedings, Part {I}}, volume 7965 of \emph{Lecture Notes in Computer
  Science}, pages 135--147. Springer, 2013.
\newblock \doi{10.1007/978-3-642-39206-1\_12}.
\newblock URL \url{https://doi.org/10.1007/978-3-642-39206-1\_12}.

\bibitem[Bienkowski et~al.(2014)Bienkowski, Byrka, Chrobak, Jez, Nogneng, and
  Sgall]{BetterApproximationBoundsfortheJointReplenishmentProblem}
Marcin Bienkowski, Jaroslaw Byrka, Marek Chrobak, Lukasz Jez, Dorian Nogneng,
  and Jir{\'{\i}} Sgall.
\newblock Better approximation bounds for the joint replenishment problem.
\newblock In Chandra Chekuri, editor, \emph{Proceedings of the Twenty-Fifth
  Annual {ACM-SIAM} Symposium on Discrete Algorithms, {SODA} 2014, Portland,
  Oregon, USA, January 5-7, 2014}, pages 42--54. {SIAM}, 2014.
\newblock \doi{10.1137/1.9781611973402.4}.
\newblock URL \url{https://doi.org/10.1137/1.9781611973402.4}.

\bibitem[Bienkowski et~al.(2016)Bienkowski, B{\"{o}}hm, Byrka, Chrobak,
  D{\"{u}}rr, Folwarczn\'y, Jez, Sgall, Nguyen, and
  Vesel{\'{y}}]{Online_Algorithms_for_Multi_Level_Aggregation}
Marcin Bienkowski, Martin B{\"{o}}hm, Jaroslaw Byrka, Marek Chrobak, Christoph
  D{\"{u}}rr, Luk\'a\v{s} Folwarczn\'y, Lukasz Jez, Jiri Sgall, Kim~Thang
  Nguyen, and Pavel Vesel{\'{y}}.
\newblock Online algorithms for multi-level aggregation.
\newblock In Piotr Sankowski and Christos~D. Zaroliagis, editors, \emph{24th
  Annual European Symposium on Algorithms, {ESA} 2016, August 22-24, 2016,
  Aarhus, Denmark}, volume~57 of \emph{LIPIcs}, pages 12:1--12:17. Schloss
  Dagstuhl - Leibniz-Zentrum f{\"{u}}r Informatik, 2016.
\newblock \doi{10.4230/LIPIcs.ESA.2016.12}.
\newblock URL \url{https://doi.org/10.4230/LIPIcs.ESA.2016.12}.

\bibitem[Bienkowski et~al.(2018)Bienkowski, Kraska, and
  Schmidt]{BienkowskiKS2018}
Marcin Bienkowski, Artur Kraska, and Pawel Schmidt.
\newblock Online service with delay on a line.
\newblock In \emph{{SIROCCO}}, 2018.

\bibitem[Bienkowski et~al.(2022)Bienkowski, B{\"{o}}hm, Byrka, and
  Marcinkowski]{OnlineFacilityLocationwithLinearDelay}
Marcin Bienkowski, Martin B{\"{o}}hm, Jaroslaw Byrka, and Jan Marcinkowski.
\newblock Online facility location with linear delay.
\newblock In Amit Chakrabarti and Chaitanya Swamy, editors,
  \emph{Approximation, Randomization, and Combinatorial Optimization.
  Algorithms and Techniques, {APPROX/RANDOM} 2022, September 19-21, 2022,
  University of Illinois, Urbana-Champaign, {USA} (Virtual Conference)}, volume
  245 of \emph{LIPIcs}, pages 45:1--45:17. Schloss Dagstuhl - Leibniz-Zentrum
  f{\"{u}}r Informatik, 2022.
\newblock \doi{10.4230/LIPIcs.APPROX/RANDOM.2022.45}.
\newblock URL \url{https://doi.org/10.4230/LIPIcs.APPROX/RANDOM.2022.45}.

\bibitem[Buchbinder et~al.(2017)Buchbinder, Feldman, Naor, and
  Talmon]{CompetitiveAlgorithmforOnlineMultilevelAggregation}
Niv Buchbinder, Moran Feldman, Joseph~(Seffi) Naor, and Ohad Talmon.
\newblock \emph{O}(depth)-competitive algorithm for online multi-level
  aggregation.
\newblock In Philip~N. Klein, editor, \emph{Proceedings of the Twenty-Eighth
  Annual {ACM-SIAM} Symposium on Discrete Algorithms, {SODA} 2017, Barcelona,
  Spain, Hotel Porta Fira, January 16-19}, pages 1235--1244. {SIAM}, 2017.
\newblock \doi{10.1137/1.9781611974782.80}.
\newblock URL \url{https://doi.org/10.1137/1.9781611974782.80}.

\bibitem[Emek et~al.(2016)Emek, Kutten, and
  Wattenhofer]{Online_matching_haste_makes_waste}
Yuval Emek, Shay Kutten, and Roger Wattenhofer.
\newblock Online matching: haste makes waste!
\newblock In \emph{Proceedings of the 48th Annual {ACM} {SIGACT} Symposium on
  Theory of Computing, {STOC} 2016, Cambridge, MA, USA, June 18-21, 2016},
  pages 333--344, 2016.

\bibitem[Epstein(2019)]{On_bin_packing_with_clustering_and_bin_packing_with_delays}
Leah Epstein.
\newblock On bin packing with clustering and bin packing with delays.
\newblock \emph{CoRR}, abs/1908.06727, 2019.

\bibitem[Gupta et~al.(2022)Gupta, Kumar, and
  Panigrahi]{Caching_with_Time_Windows_and_Delays}
Anupam Gupta, Amit Kumar, and Debmalya Panigrahi.
\newblock Caching with time windows and delays.
\newblock \emph{{SIAM} J. Comput.}, 51\penalty0 (4):\penalty0 975--1017, 2022.
\newblock \doi{10.1137/20m1346286}.
\newblock URL \url{https://doi.org/10.1137/20m1346286}.

\bibitem[Irani(1991)]{irani1991two}
Sandy Irani.
\newblock Two results on the list update problem.
\newblock \emph{Information Processing Letters}, 38\penalty0 (6):\penalty0
  301--306, 1991.

\bibitem[Reingold et~al.(1994)Reingold, Westbrook, and
  Sleator]{reingold1994randomized}
Nick Reingold, Jeffery Westbrook, and Daniel~D Sleator.
\newblock Randomized competitive algorithms for the list update problem.
\newblock \emph{Algorithmica}, 11\penalty0 (1):\penalty0 15--32, 1994.

\bibitem[Sleator and Tarjan(1985)]{sleator1985amortized}
Daniel~D Sleator and Robert~E Tarjan.
\newblock Amortized efficiency of list update and paging rules.
\newblock \emph{Communications of the ACM}, 28\penalty0 (2):\penalty0 202--208,
  1985.

\bibitem[Teia(1993)]{teia1993lower}
Boris Teia.
\newblock A lower bound for randomized list update algorithms.
\newblock \emph{Information Processing Letters}, 47\penalty0 (1):\penalty0
  5--9, 1993.

\end{thebibliography}
\clearpage

\appendix
\section{The Analysis for the Algorithm for Time Windows}
\label{section_deadlines}
In this section we will prove Theorem \ref{deadlines.thm.1}. Throughout we will denote Algorithm \ref{deadlines.alg} as $ALG$.
\begin{definition}
    For each $k\in[m]$ we use ${a_k}$ and ${q_k}$ to denote the arrival time and deadline of the request $r_k$. 
\end{definition}
As a first step towards proving Theorem \ref{deadlines.thm.1} we prove in Lemma \ref{deadlines.lemma.22} that it is enough to consider inputs that only contain triggering requests. 
\begin{lemma}
\label{deadlines.lemma.22}
Let \(\sigma\) be a sequence of requests and let \(\sigma^{'}\) be \(\sigma\) after omitting all the non-triggering requests (with respect to \(ALG\)). Then \[\frac{ALG(\sigma)}{{OPT}(\sigma)}\leq \frac{ALG(\sigma^{'})}{{OPT}(\sigma^{'})}.\]
\end{lemma}
\begin{proof}
We have that \({OPT}(\sigma^{'})\leq {OPT}(\sigma)\), since \(\sigma\) contains all the requests in \(\sigma^{'}\).
Therefore, in order to prove the lemma, it it is sufficient to prove that \(ALG(\sigma^{'})=ALG(\sigma)\). In order to prove that, it is sufficient to prove that \(ALG\) behaves in the same way for the two sequences. Indeed, assume that \(r\) is a non-triggering request in \(\sigma\) for an element \(e\) which \(ALG\) served at time \(t\). Let \(x\) be the position of \(e\) in \(ALG\)'s list at time \(t\). Let \(r^*\) be the trigger request in \(\sigma\) for an element \(e^*\) which \(ALG\) served at time \(t\). Observe that such \(r^*\) must exist. Note that \(r^*\) and \(r\) are different requests because \(r^*\) is a trigger request and \(r\) is not. Let \(x^*\) be the position of \(e^*\) in \(ALG\)'s list at time \(t\). Since \(r^*\) is a trigger request which \(ALG\) served at time \(t\), we have that \(t\) is the deadline of \(r^*\).
There are three possible cases: (1) \(e=e^*\), (2) \(e\neq e^*\) and the deadline of \(r\) is exactly \(t\) which implies that \(x<x^*\) (3) \(e\neq e^*\) and the deadline of \(r\) is after time \(t\) which implies  that \(x\leq 2x^*-1\).
It is easy to verify that for each of the cases above, omitting \(r\) from \(\sigma\) does not change the behavior of \(ALG\) on the sequence. This holds for any non-triggering request \(r\) in \(\sigma\) and thus \(ALG\) behaves on \(\sigma^{'}\) in the same as its behavior on \(\sigma\).
\end{proof}
\begin{corollary}
\label{deadlines.lemma.3}

We may assume w.l.o.g. that the input $\sigma$ only contains triggering requests (with respect to $ALG$).

\end{corollary}




\noindent The following lemma is simple but will be very useful later. Recall that $k$ refers to the index of the $k$'th request in the input $\sigma$ and that $e_k$ denotes its requested element.

\begin{lemma}
\label{deadlines.lemma.5}

For every $k \in [m]$, the position of $e_k$ in $ALG$'s list remains unchanged throughout the time interval $[a_k, q_k)$. Hence $x_k$ denotes the location of $e_k$ in $ALG$s list during the time interval $[a_k,q_k)$.

\end{lemma}
\begin{proof}
During the time interval between \(a_k\) and \(q_k\), \(ALG\) did not access the element \(e_k\) and did not access any element located after \(e_k\) in its list, because otherwise it would be a contradiction to our assumption that \(r_k\) is a trigger request. Therefore, during the time interval mentioned above, \(ALG\) only accessed (served) elements which were before \(e_k\) in its list and performed move-to-fronts on them. But these move-to-fronts did not change the position of \(e_k\) in \(ALG\)'s list.
\end{proof}




\begin{definition}
Let $z_k$ denote the position of the farthest element $OPT$ accesses at the time it served $r_k$.
\end{definition}
Note that $z_k$ defines the cost $OPT$ pays for serving the set of requests that contain $r_k$. 

Recall that \(ALG\) serves all requests separately (since all requests are triggering requests - Corollary \ref{deadlines.lemma.3}). \(OPT\), on the other hand, may serve multiple requests simultaneously. Note that at the time \(OPT\) serves the request \(r_k\), it pays access cost of \(z_k\) (and it is guaranteed that \(z_k\geq y_k\)). The strict inequality \(z_k> y_k\) occurs in case \(OPT\) serves a request for an element located further than \(e_k\) in its list and by accessing this far element, \(OPT\) also accesses \(e_k\), thus serving \(r_k\).
\begin{lemma}
\label{deadlines.lemma.6}
The cost of \(ALG\) is bounded by
\[ALG(\sigma)\leq 3\cdot\sum_{k=1}^{m}x_k\]
\end{lemma}
\begin{proof}
For each \(k\in[m]\), \(e_k\) is located at position \(x_k\) at the time when \(ALG\) serves \(r_k\). \(ALG\) pays an access cost of at most \(2x_k-1\) when it serves the request \(r_k\) (Observe that \(ALG\) may pay an access cost of less than \(2x_k-1\) in case \(n<2x_k-1\)). \(ALG\) also pays a cost of \(x_k-1\) for performing move-to-front on \(e_k\). Therefore, \(ALG\) suffers a total cost of at most \((2x_k-1)+(x_k-1)\leq 3x_k\) for serving this request. If we sum for all the requests, we get that \(ALG(\sigma)\leq 3\cdot\sum_{k=1}^{m}x_k\).
\end{proof}
\begin{lemma}
\label{deadlines.lemma.17}
Let \(t\) be a time when the active requests (indices) in \(ALG\) are \(R=\{k_{1},k_{2},...,k_{d}\}\) where \(1\leq x_{k_{1}}< x_{k_{2}}< ... < x_{k_{d}}\leq n\). We have:
\begin{enumerate}
\item For each \(\ell\in[d-1]\), we have \(q_{k_{{\ell}}}\leq q_{k_{{\ell+1}}}\), i.e., \(ALG\) serves the request \(r_{k_{{\ell}}}\) before it serves  \(r_{k_{{\ell+1}}}\).
\item For each \(\ell\in[d-1]\) we have that \(2x_{k_{\ell}}\leq x_{k_{{\ell+1}}}\).
\item \(d\leq{\log n}+1\), i.e., at any time, there are at most \({\log n}+1\) active requests in \(ALG\).
\end{enumerate}
\end{lemma}
\begin{proof}
If \(ALG\) serves \(r_{k_{{\ell+1}}}\) before it serves \(r_{k_{{\ell}}}\), it serves also \(r_{k_{{\ell}}}\) by passing through 
\(e_{k_{{\ell}}}\) when it accesses \(e_{k_{{\ell+1}}}\), contradicting our assumption that \(r_{k_{{\ell}}}\) is a trigger request (Observation \ref{deadlines.lemma.3}).

If we have \(x_{k_{{\ell+1}}}\leq 2x_{k_{\ell}}-1\), then when \(ALG\) serves \(r_{k_{{\ell}}}\), it will also access \(e_{k_{{\ell+1}}}\), thus serving \(r_{k_{{\ell+1}}}\), contradicting our assumption that \(r_{k_{{\ell+1}}}\) is a trigger request (Observation \ref{deadlines.lemma.3}).

Using what we have already proved, a simple induction can be used in order to prove that for each \(\ell\in[|R^{OPT}_t|]\), we have that \(2^{\ell-1}x_{k^t_{1}}\leq x_{k^t_{\ell}}\). Therefore, we have that \(2^{|R^{OPT}_t|-1}x_{k^t_{1}}\leq x_{k^t_{|R^{OPT}_t|}}\). We also have that \(1\leq x_{k^t_{1}}\) and \(x_{k^t_{|R^{OPT}_t|}}\leq n\). These three inequalities yield to \(|R^{OPT}_t|\leq{\log n}+1\).
\end{proof}
\noindent Next we consider $OPT$'s solution.

\begin{definition}
Let \(T_{OPT}\) be the set of times when \(OPT\) served requests. We then define:
\noindent\begin{itemize}
\item For each time \(t\in T_{OPT}\), let \(R^{OPT}_t=\{k^t_1,k^t_2,...,k^t_{|R^{OPT}_t|}\}\) be the non-empty set of request indices that \(OPT\) served at time \(t\) where \(1\leq x_{k^t_1}< x_{k^t_2}< ...< x_{k^t_{|R^{OPT}_t|}}\leq n\).
\item Let \(J(t)=\underset{k\in {R^{OPT}_t}}{\arg\max}\{y_{k}\}\).
\end{itemize}
\end{definition}

For a pictorial example see Figure \ref{fig.alg_behavior_with_distances} below and Figure \ref{fig.alg_behavior_actual_list} at the end of this section. By definition for each \(t\in T_{OPT}\), we have\[1\leq y_{J(t)}=z_{k^t_1}=z_{k^t_2}=...=z_{k^t_{|R^{OPT}_t|-1}}=z_{k^t_{|R^{OPT}_t|}}\leq n\]
Observe that at time \(t\in{T_{OPT}}\), \(OPT\) serves the requests \(R^{OPT}_t\) together by accessing the \(y_{J(t)}\)'s element in its list. Therefore, \(OPT\) pays an access cost of \(y_{J(t)}\) at time \(t\).
\begin{observation}
\label{deadlines.lemma.7}
For any \(t\in{T_{OPT}}\), the total cost \(OPT\) pays for accessing elements at time \(t\) is \(y_{J(t)}\).
\end{observation}
\begin{lemma}
\label{deadlines.lemma.8}
Let \(t\in{T_{OPT}}\). We have:
\begin{enumerate}
\item For each \(\ell\in[|R^{OPT}_t|-1]\) we have that \(2x_{k^t_\ell}\leq x_{k^t_{\ell+1}}\).
\item \(|R^{OPT}_t|\leq{\log n}+1\), i.e, \(OPT\) serves at most \({\log n}+1\) triggers at the same time.
\end{enumerate}
\end{lemma}
\begin{proof}
Since \(OPT\) served the requests \(R^{OPT}_t\) at time \(t\), all these requests arrived at time \(t\) or before it. Therefore, from Observation \ref{deadlines.lemma.4} we get that all the requests \(R^{OPT}_t\) were active in \(ALG\) at time \(t\). Therefore, we get that this lemma holds due to Lemma \ref{deadlines.lemma.17}. Note that there may be additional requests which were active in \(ALG\) at time \(t\) but \(OPT\) did not serve at time \(t\) (meaning they were not in \(R^{OPT}_t\)), but this does not contradict the conclusion.
\end{proof}
The following lemma allows us to consider from now on only algorithms such that if they serve requests at time \(t\), then at least one of these requests has a deadline at \(t\). In particular, we can assume that \(OPT\) has this property. Observe that \(ALG\) also has this property.
\begin{lemma}
\label{deadlines.lemma.1}
For every algorithm \(A\), there exists an algorithm \(B\) such that for each sequence of request \(\sigma\) we are guaranteed that:
\begin{enumerate}
\item \(B\) only serves requests upon some deadline.
\item \(B(\sigma)\leq A(\sigma)\).
\end{enumerate}
\end{lemma}
\begin{proof}
Given the algorithm \(A\), we define the algorithm \(B\) as follows. At a time \(t\) when \(A\) serves requests with deadlines greater than \(t\), \(B\) is defined not to serve any request. Instead, \(B\) waits and delays the serving of these active requests. \(B\) keeps delaying the serving of requests that \(A\) is serving until one of the delayed requests reaches its deadline - only then \(B\) serves all the active requests that have been delayed so far. \(B\) also changes its list so it will match \(A\)'s current list.

Clearly, \(B\) is a valid algorithm, meaning that \(B\) serves each request during the time interval between the request's arrival and its deadline. This follows from the definition of \(B\) and the assumption that \(A\) is a valid algorithm. We have that the total cost \(B\) pays for accessing elements is not bigger than the total cost \(A\) pays for accessing elements because by serving more requests together - the total access cost can only decrease. We also have that the number of paid swaps \(A\) does is the same as the number of swaps \(B\) does. Therefore we have that \(B(\sigma)\leq A(\sigma)\). Moreover, due to the definition of \(B\), it is guaranteed that at each time \(B\) serves requests - at least one of them reached its deadline.
\end{proof}
\begin{figure}[t]
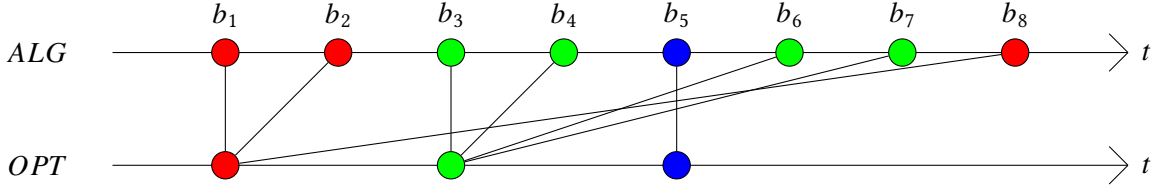

\centering
\ctikzfig{figs/deadlineFig}

\caption{An example of \(ALG\)'s behavior compared to \(OPT\)'s behavior on a sequence \(\sigma\) during the time horizon. The horizontal lines are the time horizon of \(ALG\) and \(OPT\) running on \(\sigma\). For each request (trigger) in \(\sigma\) - there is a dot in the top horizontal line at the time \(ALG\) served this request which is its deadline  (Corollary \ref{deadlines.lemma.3}). Therefore, the requests (dots) in the top horizontal line are ordered according to their deadlines (which may be different from their arrival times order). The dots in the bottom line are at times when \(OPT\) served requests. For each request (trigger) in \(\sigma\) there is a non-horizontal segment between the time when \(ALG\) served this request and the time when \(OPT\) served this request. \(OPT\) served first the 3 red dot requests (together), then the 4 gre`en dot requests (together) and then the blue dot request. By Lemma \ref{deadlines.lemma.1}, each  monochromatic requests are served by \(OPT\) at their minimum deadline. Denote by \(b_k\) the position of the requested element in \(ALG\)'s list. By Lemma \ref{deadlines.lemma.8} we have \(4b_1\leq 2b_2\leq b_8\) and \(8b_3\leq 4b_4 \leq 2b_6 \leq b_7\). By Lemma \ref{deadlines.lemma.17} we have that \(8b_5\leq 4b_6 \leq 2b_7 \leq b_8\) since the 4  rightmost dots requests in the top horizontal line were active in \(ALG\) at the deadline of the blue request.}
\label{fig.alg_behavior_with_distances}
\end{figure}

For convenience, we assume that when both \(ALG\) and \(OPT\) are serving \(\sigma\), in case both \(OPT\) and \(ALG\) perform access or swapping operations at the same time - we first let \(OPT\) perform its operations and only then \(ALG\) will perform its operations. 

On the other hand, for elements which are not served at the same time by $OPT$ and $ALG$, by combining the fact that $ALG$ serves requests at their deadline (see Corollary \ref{deadlines.lemma.3}) with the fact  that $OPT$ must serve requests before the deadline, we get that again $OPT$ serves the request before $ALG$. Combining the two cases yields Observation \ref{deadlines.lemma.4}.


\begin{observation}
\label{deadlines.lemma.4}
For each \(k\in[m]\), \(OPT\) serves the request \(r_k\) before \(ALG\) serves \(r_k\).
\end{observation}
\begin{definition}
We define the set of \textbf{events} \({P}\) which contains the following \(3\) types of events:
\begin{enumerate}
\item \(ALG\) serves the request \(r_k\) at time \(q_k\).
\item \(OPT\) serves the requests \(R^{OPT}_t\) at time \(t\).
\item \(OPT\) swaps two elements.
\end{enumerate}
\end{definition}
Recall that the potential function $\Phi$ is defined in Section \ref{section_potential_function_deadlines} as follows:
\[\Phi(t)=4\cdot \phi(t)+\sum_{k\in\lambda(t)}\psi(x_k,y_k)+4\cdot\sum_{k\in\lambda(t)}\mu_k(t)\]
where the terms $\phi$, $\lambda$, $\psi$ and $\mu_k$ are also defined in that section.
\begin{definition}
For each event \(p\in P\), we define:
\begin{itemize}
\item \(ALG^p\) (\(OPT^p\)) to be the cost \(ALG\) (\(OPT\)) pays during \(p\).
\item For any parameter $z$, $\Delta z^p$ to be the value of $z$ after $p$ minus the value of $z$ before $p$.
\end{itemize}
\end{definition}
Clearly, we have \(ALG(\sigma)=\sum_{p\in P}ALG^p\) and
\(OPT(\sigma)=\sum_{p\in P}OPT^p\). Observe that \(\Phi\) starts with \(0\) (since at the beginning, the lists of \(ALG\) and \(OPT\) are identical) and is always non-negative. Therefore, if we prove that for each event \(p\in P\), we have
\[ALG^p+\Delta\Phi^p\leq 24\cdot OPT^p\]
then, by summing it up for over all the events, we will be able to prove Theorem \ref{deadlines.thm.1}. Note that we do not care about the actual value \(\Phi(t)\) by itself, for any time \(t\). We will only measure the change of \(\Phi\) as a result of each type of event in order to prove that the inequality mentioned above indeed holds. The three types of events that we will discuss are:
\begin{enumerate}
\item The event where \(ALG\) serves the request \(r_k\) at time \(q_k\) (event type 1) is analyzed in Lemma \ref{deadlines.lemma.10}.
\item The event where \(OPT\) serves the requests \(R^{OPT}_t\) at time \(t\) (event type 2) is analyzed in Lemma \ref{deadlines.lemma.14}.
\item The event where \(OPT\) swaps two elements (event type 3) is analyzed in Lemma \ref{deadlines.lemma.16}.
\end{enumerate}
We begin by analyzing the event where \(ALG\) served a request.
\begin{lemma}
\label{deadlines.lemma.10}
Let \(p\in P\) be the event where \(ALG\) served the request \(r_k\) (where \(k\in[m]\)) at time \(q_k\). We have
\[ALG^p+\Delta\Phi^p\leq 0\ (=OPT^p)\]
\end{lemma}
In order to prove Lemma \ref{deadlines.lemma.10}, we separate the movement of $ALG$ versus the movement of $OPT$. The final proof is the superposition of the two movements. Firstly we assume that $OPT$ did not increase the position of $e_k$ in its list ever since it served the request $r_k$ until $ALG$ served it, then we remove this assumption.
\begin{lemma}
\label{deadlines.lemma.18}
Let \(p\in P\) be the event where \(ALG\) served the request \(r_k\) (where \(k\in[m]\)) at time \(q_k\). Assume that ever since \(OPT\) served \(r_k\) until \(ALG\) served \(r_k\), \(OPT\) did not increase the position of \(e_k\) in its list. We have that
\[3x_k+4\cdot\Delta\phi^p-\psi(x_k,y_k)\leq 0\]
\end{lemma}
\begin{proof}
The assumption means that at time \(q_k\), the position of \(e_k\) in \(OPT\)'s list is at most \(y_k\) (it may be even lower, due to movements which may be performed by \(OPT\) to \(e_k\) towards the beginning of its list, after \(OPT\) served \(r_k\)). Recall that after \(ALG\) serves \(r_k\), the position of \(e_k\) in \(ALG\)'s list changes from \(x_k\) to $1$, as a result of the move-to-front \(ALG\) performs on \(e_k\). In order to prove the required inequality, we consider the following cases, depending on the value of \(y_k\):
\begin{itemize}
\item The case \(1\leq x_k\leq y_k\). We have \(\psi(x_k,y_k)=7x_k\).

Therefore, observe that it is sufficient to prove that \(\Delta\phi^p\leq x_k\).

This is indeed the case, because moving \(e_k\) from position \(x_k\) to position \(1\) in \(ALG\)'s list required \(ALG\) to perform \(x_k-1\) swaps, each one of those caused \(\phi\) to either increase by \(1\) or decrease by \(1\). Therefore, all these \(x_k-1\) swaps cause \(\phi\) to increase by at most \(x_k-1\).
\item The case \(y_k\leq x_k\leq n\).

On one hand, there are at least \(x_k-y_k\) elements which were before \(e_k\) in \(ALG\)'s list and after \(e_k\) in \(OPT\)'s list before the move-to-front \(ALG\) performed on \(e_k\), but they will be after \(e_k\) in \(ALG\)'s list after this move-to-front. This causes \(\phi\) to decrease by at least \(x_k-y_k\). On the other hand, there are at most \(y_k-1\) elements which were before \(e_k\) in both \(OPT\)'s list and \(ALG\)'s list before the move-to-front \(ALG\) performed on \(e_k\), but they will be after \(e_k\) in \(ALG\)'s list after this move-to-front. This causes \(\phi\) to increase by at most \(y_k-1\). Therefore, we have that \[\Delta\phi^p\leq -(x_k-y_k)+(y_k-1)=2y_k-x_k-1\]
Hence,
\begin{align*}
3x_k+4\cdot\Delta\phi^p
& \leq 3x_k+4\cdot(2y_k-x_k-1) \leq 8y_k-x_k
\end{align*}
Now we distinguish between these two following cases, depending on the value of \(y_k\):
\begin{itemize}
\item The case \(y_k\leq x_k \leq 8y_k\). We have that \(\psi(x_k,y_k)=8y_k-x_k\).
Hence
\begin{align*}
3x_k+4\cdot\Delta\phi^p-\psi(x_k,y_k) \leq 8y_k-x_k-\psi(x_k,y_k)  = 8y_k-x_k-(8y_k-x_k) = 0
\end{align*}
\item The case \(8y_k\leq x_k \leq n\). We have that \(\psi(x_k,y_k)=0\).
Hence
\begin{align*}
3x_k+4\cdot\Delta\phi^p-\psi(x_k,y_k)
 \leq 8y_k-x_k-\psi(x_k,y_k) = 8y_k-x_k \leq 8y_k-8y_k = 0
\end{align*}
\end{itemize}
\end{itemize}
\end{proof}
Now we are ready to complete the proof of Lemma \ref{deadlines.lemma.10}.
\begin{proof}[Proof of Lemma \ref{deadlines.lemma.10}]
Since $OPT$ has already served this request, we have \(OPT^p=0\). As explained in the proof of Lemma \ref{deadlines.lemma.6}, we have \(ALG^p\leq 3x_k\). Therefore, we are left with the task to prove that
\[3x_k+\Delta\Phi^p\leq 0\]
Observe that \(\psi(x_k,y_k)\) and \(\mu_k(t)\) are dropped (and thus are subtracted) from \(\Phi\) as a result of \(ALG\) serving \(r_k\). Therefore, we are left with the task to prove that
\[3x_k+4\cdot\Delta\phi^p-\psi(x_k,y_k)-4\cdot\mu_k(t)\leq 0\]
We first assume that ever since \(OPT\) served \(r_k\) until \(ALG\) served \(r_k\), \(OPT\) did not increase the position of \(e_k\) in its list (later we remove this assumption). This assumption means that \(\mu_k(t)=0\).
Therefore, due to Lemma \ref{deadlines.lemma.18}, we have that the above inequality holds. We are left with the task to prove that the above inequality continues to hold even without this assumption. 

Assume that ever since \(OPT\) served \(r_k\) until \(ALG\) served \(r_k\), \(OPT\) performed a swap between \(e_k\) and another element where \(e_k\)'s position has been increased as a result of this swap. We shall prove that the above inequality continues to hold nonetheless.

On one hand, this swap causes either an increase of \(1\)  or a decrease of \(1\) to \(\Delta\phi^p\). Therefore,
the left term of the inequality will be increased by at most \(4\). On the other hand, the left term of the inequality will certainty be decreased by \(4\) as a result of this swap, because \(\mu_k\) will certainty be increased by \(1\). To conclude, a decrease of at least \(4-4=0\) will be applied to the left term of the inequality, thus the inequality will continue to hold after this swap as well.

By using the argument above for each swap of the type mentioned above, we get that the above inequality continues to hold even without the assumption that \(OPT\) did not increase the position of \(e_k\) in its list since it served \(r_k\) until \(ALG\) served \(r_k\), thus the lemma has been proven.
\end{proof}
Now that we analyzed the event when \(ALG\) serves a request, the next target is to analyze the event where \(OPT\) serves multiple request together.
\begin{figure}[H]
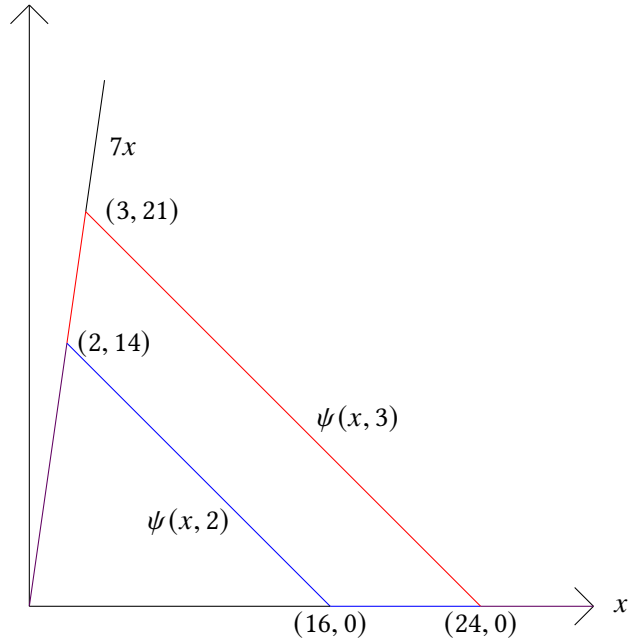

\centering
\ctikzfig{figs/psiIllustrateHelpLemma}
\caption{The graphs of \(\psi(x,2)\) (in blue and purple), \(\psi(x,3)\) (in red and purple) and the function \(7x\).}
\label{fig.illustrate_observation}
\end{figure}
The following observation contains useful properties of \(\psi\) that will be used later on. The reader may prove them algebraically. Alternatively, he can look at Figure \ref{fig.illustrate_observation}, which illustrates these claims and convince himself that they indeed hold.
\begin{observation}
\label{deadlines.lemma.15}
For each \(x,x',y,y'\in[1,\infty)\) such that \(x\leq x'\) and \(y\leq y'\), the function \(\psi\) satisfies the following claims:
\begin{enumerate}
\item \(0\leq\psi(x,y)\leq 7x\).
\item If \(y\leq x\leq x'\) then \(\psi(x',y)\leq\psi(x,y)\).
\item \(\psi(x,y)\leq\psi(x,y')\).
\end{enumerate}
\end{observation}
The target now is to analyze the event when \(OPT\) serves the requests \(R^{OPT}_t\) together at time \(t\). Recall that when \(OPT\) serves a request \(r_k\), the value \(\psi(x_k,y_k)\) is added to \(\Phi\). The following lemma will be needed in order to analyze this event. 
\begin{lemma}
\label{deadline.lemma.24}
Let \(a>0\) and
let \(f:[0,\infty)\rightarrow[0,\infty)\) be the function defined as follows: 
\[
f(x) = \begin{cases}
7x & \text{if  } 0\leq x\leq a \\
8a-x & \text{if  } a\leq x\leq 8a \\
0 & \text{if  } 8a\leq x
\end{cases}
\]
Consider the optimization problem $Q$ of choosing a (possibly infinite) subset \(U\subseteq (0,8a)\) that will maximize \(\sum_{x\in U}f(x)\) with the requirement \(\forall x,y\in U:x<y\implies 2x\leq y\).
Then the optimal value of $Q$ is $24a$.
\end{lemma}
\begin{proof}
Let $U$ be a solution of $Q$. By feasibility we have $|U\cap(\frac{1}{4}a,a)|\leq2$; we assume the intersection is not empty otherwise we could add $\frac{1}{2}a$ to $U$ and improve its value while maintaining feasibility. Let $z=\max U\cap(\frac{1}{4}a,a)$. 
We have $z\geq\frac{1}{2}a$ because otherwise we could 
replace $z$ with $\frac{1}{2}a$ and get a feasible solution with a bigger value.
Similarly, by feasibility we have $|U\cap[a,4a)|\leq2$; as before we assume it is not empty otherwise we could add $2a$ to $U$ and improve its value, while maintaining feasibility. Let $w=\min U\cap[a,4a)$. 

Note that $a\leq 2z<2a$. Since $U$ is a feasible solution and $z,w\in U$, we have $2z\leq w$. Therefore we may assume that $w=2z$ otherwise we could 
replace $w$ with $2z$ and get an increased feasible solution.

Since $U$ is a feasible solution, we must have $|U\cap[a,8a)|\leq 3$. Among these 3 numbers, the minimum number is $w$,
the middle number is at least $2w$ and the maximum number is at least twice the middle number i.e. at least $4w$.
Also, for each $i\in\mathbb{N}$, the $i$-th largest number in the set $U\cap(0,a)$ is at most $\frac{1}{2^{i-1}}z$ (by easy induction on $i$ from feasibility).
Since $f$ is monotone-increasing in the interval $[0,a]$ and monotone-decreasing in the interval $[a,8a]$, we can bound the value of $U$ 
as follows:
\begin{align*}
\sum_{x\in U}f(x) & \leq \sum_{i=1}^{\infty}f(\frac{1}{2^{i-1}}z)+f(w)+f(2w)+f(4w)\\ & = \sum_{i=1}^{\infty}7\frac{1}{2^{i-1}}z+(8a-w)+(8a-2w)+(8a-4w)\\
& = 7z\cdot\underbrace{\sum_{i=0}^{\infty}\frac{1}{2^{i}}}_{=2}+24a-7\cdot\underbrace{w}_{=2z} =24a
\end{align*}
Finally, it is easy to see that the feasible solution $\{\frac{4a}{2^{i}}|i\in\mathbb{N}\cup\{0\}\}$ has a value of $24a$ and therefore the optimal value of $Q$ is exactly $24a$. 
\end{proof}
\noindent Now we can use Lemma \ref{deadline.lemma.24} in order to analyze the event when \(OPT\) serves multiple requests together.
\begin{lemma}
\label{deadlines.lemma.14}
Let \(p\in P\) be the event where \(OPT\) served the requests \(R^{OPT}_t\) at time \(t\). We have
\[ALG^p+\Delta\Phi^p\leq 24\cdot OPT^p\]
\end{lemma}
\begin{proof}
We have \(OPT^p=y_{J(t)}\) (Observation \ref{deadlines.lemma.7}). We also have \(ALG^p=0\). Therefore, the task is to prove that \[\Delta\Phi^p\leq 24\cdot y_{J(t)}\]
Since \(OPT\) did not change its list, we have \(\Delta\phi^p=0\). For each \(\ell\in[R^{OPT}_t]\), the value \(\psi(x_{k^t_\ell},y_{k^t_\ell})\) is added to \(\Phi\). No other changes are applied to \(\Phi\) as a result of \(OPT\) serving the requests \(R^{OPT}_t\). Let $X=\{x\in[8y_{J(t)}-1]|\exists\ell\in [R^{OPT}_t]:x_{k^t_\ell}=x\}$ We have
\begin{align*}
\Delta\Phi^p
& = \sum_{\ell=1}^{|R^{OPT}_t|}\psi(x_{k^t_\ell},y_{k^t_\ell}) \leq \sum_{\ell=1}^{|R^{OPT}_t|}\psi(x_{k^t_\ell},y_{J(t)}) \\
& = \sum_{\ell\in [R^{OPT}_t]:x_{k^t_\ell}\in[8y_{J(t)}-1]}\psi(x_{k^t_\ell},y_{J(t)}) = \sum_{x\in X}\psi(x,y_{J(t)})\leq 24\cdot y_{J(t)}
\end{align*}
where the first inequality is due to Observation \ref{deadlines.lemma.15} (part 3) and the fact that \(y_{k^t_l}\leq z_{k^t_l}=y_{J(t)}\). The second equality is because \(\psi(x,y_{J(t)})=0\) for each \(x\geq 8y_{J(t)}\). We are left with the task of explaining the second inequality. 

Consider the optimization problem \(Q^{'}\) where \(y_{J(t)}\) is fixed and we need to choose a subset \(X\subseteq[8y_{J(t)}-1]\) that will maximize the term \(\sum_{x\in X}\psi(x,y_{J(t)})\) with the requirement \(\forall x_1,x_2\in X: x_1<x_2 \implies 2x_1\leq x_2\). This requirement must hold due to Lemma \ref{deadlines.lemma.8} (part 1). It can be seen that each feasible solution of the optimization problem \(Q^{'}\) is also a feasible solution of the optimization problem \(Q\) discussed in Lemma \ref{deadline.lemma.24} (where \(a=y_{J(t)}\) and the function \(f\) is \(f(x)=\psi(x,a)\)). Of course, there are feasible solutions of \(Q\) which are not feasible solutions of \(Q^{'}\): these are the feasible solutions of \(Q\) which contain non-integer values (and in particular, the feasible solutions of \(Q\) which are infinite sets). Also, when choosing values for \(X\) in \(Q^{'}\), it is required not choose values which are greater than \(n\) - a constraint which is not present in \(Q\). To conclude, since each feasible solution of $Q^{'}$ is a feasible solution of $Q$, the optimal value of \(Q^{'}\) is bounded by the optimal value of \(Q\), which is \(24y_{J(t)}\) according to Lemma \ref{deadline.lemma.24}. This explains the second inequality.
\end{proof}
\begin{figure}[t]
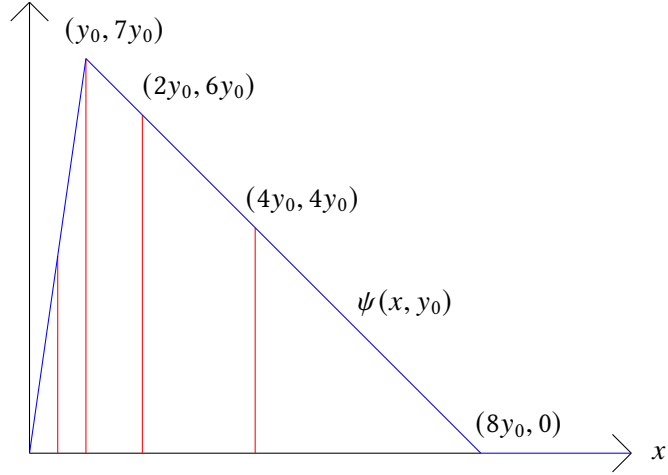

\centering
\ctikzfig{figs/psiBoundPotentialIncreaseLemma}
\caption{The usage of Lemma \ref{deadline.lemma.24} and Lemma \ref{deadlines.lemma.14}.}
\label{fig.bound_lemma}
\end{figure}
In Figure \ref{fig.bound_lemma}, we see how we got the bound of \(\sum_{l=1}^{|R^{OPT}_t|}\psi(x_{k^t_l},y_{k^t_l})\) in Lemma \ref{deadline.lemma.24} and Lemma \ref{deadlines.lemma.14}. For a fixed \(y_0\in\mathbb{N}\), we have the graph \(\psi(x,y_0)\) in blue. The red segments correspond to the \(x\) values chosen in order to bound the term \(\sum_{l\in [R^{OPT}_t]:1\leq x_{k^t_l}\leq n}\psi(x_{k^t_l},y_{k^t_l})\) in Lemma \ref{deadline.lemma.24}. There are additional red segments which should have been included between the most left one and the line $x=0$ but they are not included in this figure.
\begin{figure}[t]
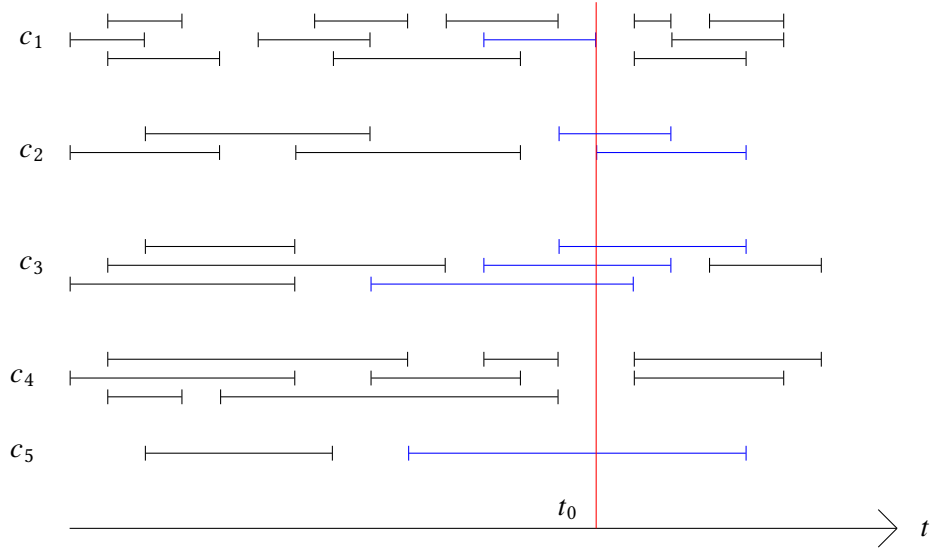

\centering
\ctikzfig{figs/timeWindowsExample}
\caption{An example of a sequence of requests. For each \(1\leq i\leq 5\), we see the time windows of the requests for the element \(c_i\) during the time horizon. Each request is represented by a segment which is the request's time window. At time \(t=t_0\) (which is represented by a vertical red line) - an algorithm may serve requests for the elements \(c_1\), \(c_2\), \(c_3\) and \(c_5\) together: these requests are blue in this figure. One of these requests (the request for the element \(c_1\)) just reached its deadline at time \(t_0\), i.e. its time window ends. Another request (a request for the element \(c_2\)) just arrived at time \(t_0\), i.e. its time window begins. Note that, as seen in the example, we may have multiple requests on a single element at a given moment.}
\label{fig.timeWindowsExample}
\end{figure}
We have analyzed the event where \(ALG\) serves a request and the event when \(OPT\) serves multiple requests together. The only event which is left to be analyzed is the event when \(OPT\) performs a swap. We analyze it below.
\newline
\begin{lemma}
\label{deadlines.lemma.16}
Let \(p\in P\) be the event where \(OPT\) performed a swap at time $t$. We have
\[ALG^p+\Delta\Phi^p\leq 8\cdot OPT^p\]
\end{lemma}
\begin{proof}
Let us assume \(OPT\) performs the swap between two elements \(i\) and \(j\), where \(j\) was the next element after \(i\) in \(OPT\)'s list prior to this swap. We have \(ALG^p=0\), \(OPT^p=1\). Therefore, the target is to prove that
\[\Delta\Phi^p\leq 8\]
The swap causes \(\phi\) to either increase by \(1\) (in case \(i\) is before \(j\) in \(ALG\)'s list at time \(t\)) or to decrease by \(1\) (otherwise). Therefore, we have \(\Delta\phi^p\leq 1\). In case there is an active request \(r_k\) for \(i\) in \(ALG\) which has already been served by \(OPT\) but has not been served by \(ALG\) yet - \(\mu_k\) will be increased by \(1\) as well. Therefore we have \[\Delta\Phi^p\leq 4\cdot 1+0+4\cdot 1=8\]
Observe that the amount of active requests for \(i\) which have been served by \(OPT\) but haven't been served by \(ALG\) is at most \(1\): if there were \(2\) (or more) such requests, at least one of them would be a non-triggering request, contradicting Corollary \ref{deadlines.lemma.3}. Therefore, the reader can verify that there are no changes to \(\Phi\) as a result of this swap that \(OPT\) performed, other than the changes mentioned above.
\end{proof}

\noindent We are now ready to prove Theorem \ref{deadlines.thm.1}.

\begin{proof}[Proof of Theorem \ref{deadlines.thm.1}]
Due to Lemma \ref{deadlines.lemma.10}, Lemma \ref{deadlines.lemma.14} and Lemma \ref{deadlines.lemma.16}, we have for each event \(p\in P\) that
\[ALG^p+\Delta\Phi^p\leq 24\cdot OPT^p\]
The theorem follows by summing it up for over all events and use the fact that \(\Phi\) starts with \(0\) and is always non-negative.
\end{proof}

\begin{figure}
\centering
\includegraphics[width= 14cm]{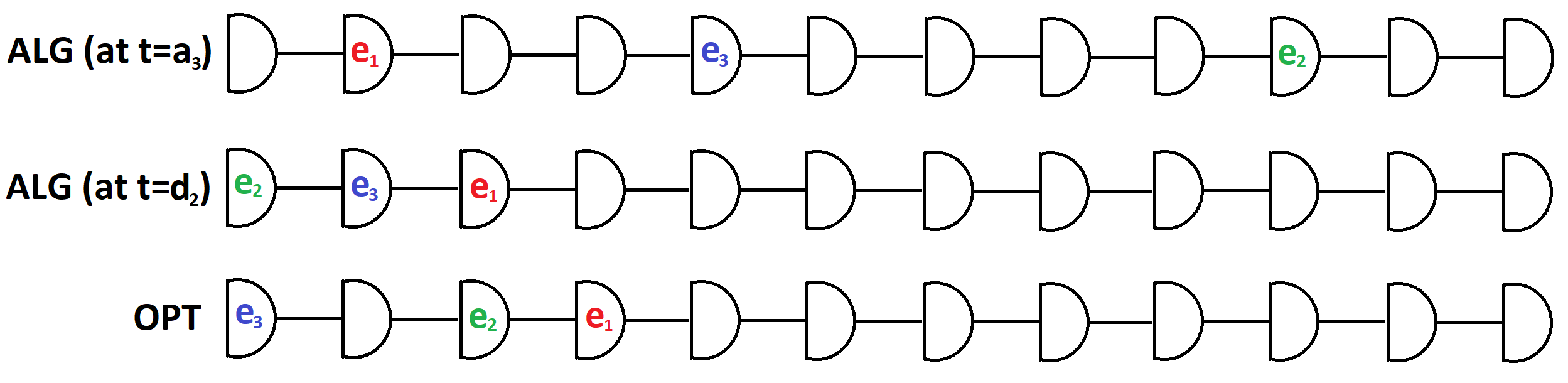}
\ctikzfig{figs/timeWindows3RequestsExample}
\caption{Another example of \(ALG\)'s behavior compared to \(OPT\)'s behavior on a sequence \(\sigma\) where \(OPT\) served all of them together, this time the time windows of the requests also appear. The trigger requests \(r_1\), \(r_2\) and \(r_3\) were served separately by \(ALG\): the top list is \(ALG\)'s list before \(ALG\) served these requests, the middle list is \(ALG\)'s list after \(ALG\) served these requests. We have \(x_1=2,x_2=10,x_3=5\). At time \(q_1\), \(ALG\) served \(r_1\): it paid an access cost of \(3\) and a swapping cost of \(1\) for moving \(e_1\) to the beginning of its list. Then, at time \(q_3\), \(ALG\) served \(r_3\): it paid an access cost of \(9\) and a swapping cost of \(4\) for moving \(e_3\) to the beginning of its list. Then, at time \(q_2\), \(ALG\) served \(r_2\): it paid an access cost of \(12\) and a swapping cost of \(9\) for moving \(e_2\) to the beginning of its list. Therefore we have \(A(\sigma)=3+1+9+4+12+9=38\). The bottom list is \(OPT\)'s list. We assume that \(OPT\) served the \(3\) requests together without performing any swaps. 
Due to Lemma \ref{deadlines.lemma.1}, we have that \(OPT\) served these requests at time \(q_1\), which is the earliest deadline of the requests. Therefore, we have that \(T_{OPT}=\{q_1\}\) and \(R^{OPT}_{q_1}=\{1,2,3\}=\{k^{q_1}_1,k^{q_1}_2,k^{q_1}_3\}\) where \(k^{q_1}_1=1,k^{q_1}_2=3,k^{q_1}_3=2\). We have \(y_1=4,y_2=3,y_3=1\). Therefore we have \(J(q_1)=1\) and \(z_1=z_2=z_3=y_1=4\). We have \(OPT(\sigma)=y_{J(q_1)}=4\).}
\label{fig.alg_behavior_actual_list}
\end{figure}
\section{The Analysis for the Algorithm for Delay}
\label{section_delay}
In addition to the terms defined in Section \ref{section_potential_function_delay}, we will use the following as well.
\begin{definition}
For time $t$, element $e\in\mathbb{E}$, request index $k\in[m]$ and position in the list $\ell\in[n]$ we define:
\begin{itemize}
    \item $RC_k^t$ to be the value of the request counter $RC_k$ at time $t$.
    \item $loc_\ell^t\in\mathbb{E}$ to be the element located in position $\ell$ at time $t$.
    \item $NI_e^t=\{i\in\mathbb{E}|i \text{ is before } e \text{ in } ALG\text{s list and before } e \text{ in } OPT\text{s list at time } t\}$.
\end{itemize}
\end{definition}
In Section \ref{section_potential_function_delay} we defined when a request is considered active in $ALG$ and when it is considered frozen in $ALG$. Now we also define the following.\footnote{Clearly, $r_k$ is frozen in $ALG$ if and only if it is frozen with $RC_k$ \textbf{or} frozen without $RC_k$ in $ALG$.}
\begin{definition}
For each $k\in[m]$, the request $r_k$ is considered:
\begin{itemize}
    \item \textbf{frozen with $RC_k$} from the time it is served by $ALG$ until the request counter $RC_k$ is deleted by $ALG$.
    \item \textbf{frozen without $RC_k$} from the time $RC_k$ is deleted by $ALG$ until ${EC}_{e_k}$ is zeroed in an element counter event of $e_k$.
    \item \textbf{deleted} after the element counter event for $e_k$ mentioned above occurs.
\end{itemize}
\end{definition}
\begin{observation}
For each $k\in[m]$, at  a time $t$ when the request $r_k$ is active in $ALG$, we have 

$RC_k^t=d_k(t)\leq d_k$.
\end{observation}
\begin{observation}
\label{delay.lemma.dk_frozen}
For each $k\in[m]$ we have $d_k(t)=d_k$ for each time $t$ in which $r_k$ is frozen in $ALG$ (And in particular, if $r_k$ is frozen with $RC_k$ in $ALG$ at time $t$ then we have $RC_k^t=d_k(t)=d_k$). ${EC}_{e_k}$ does not increase either during the time when $r_k$ is frozen, unless requests for $e_k$ arrive after $r_k$ has been served by $ALG$ (i.e. after it has became frozen) and these requests suffer delay penalty.
\end{observation}
\begin{lemma}
\label{delay.lemma.xk_lemma}
Assume the request $r_k$ (where $k\in[m]$) has been arrived at time $t$ and has been served by $OPT$ at time $t^{'}$ (where $t^{'}\geq t$). At time $t^{''}>t^{'}$ we have:
\begin{itemize}
    \item $x_{e_k}^{t^{''}}=x_k$ if $r_k$ was active in $ALG$ at time $t^{''}$.
    \item $x_{e_k}^{t^{''}}\geq x_k$ if $r_k$ was frozen in $ALG$ at time $t^{''}$.
\end{itemize}
\end{lemma}
\begin{proof}
\textbf{Proof of part 1}: Due to the  behavior of the algorithm, if $ALG$ had performed move-to-front on $e_k$ or on an element located after it in its list between time $t$ and time $t^{''}$, then $ALG$ would have also accessed $e_k$ in this operation, contradicting our assumption that $r_k$ was active in $ALG$ at time $t^{''}$: $ALG$ did not serve $r_k$ (i.e. did not access $e_k$) from time $t$ until time $t^{''}$. Therefore $e_k$ and all the elements located after it in $ALG$s list remained in the same position between time $t$ and time $t^{''}$. From time $t$ until time $t^{''}$, $ALG$ could perform move-to-front on elements located before $e_k$ in its list - but this did not affect the position of $e_k$ and the elements located after it in $ALG$s list. Therefore we have $x_{e_k}^{t^{''}}=x_{e_k}^{t^{'}}=x_{e_k}^{t}$. Recall that we have $x_k=x_{e_k}^{t^{'}}$ and that proves the first part of the lemma.

\textbf{Proof of part 2}: From the time $ALG$ served $r_k$ (i.e the time $r_k$ stopped being active in $ALG$ and started being frozen in $ALG$) until time $t^{''}$, the position of $e_k$ in $ALG$s list could increase due to move-to-fronts performed by $ALG$ to elements located after $e_k$ in its list. During this time, the position of $e_k$ in $ALG$s list could not decrease because the only way for that to happen is due to an element counter event on $e_k$. But that means $r_k$ would have been deleted in $ALG$ by time $t^{''}$, contradicting our assumption that $r_k$ is still frozen in $ALG$ at time $t^{''}$.
\end{proof}
At each time $t$, for each element $e\in\mathbb{E}$, there are $x_e^t-1$ elements which are located before $e$ in $ALG$s list at time $t$. Each one of them either located before $e$ in $OPT$s list at time $t$ (i.e. is in $NI_e^t$) or located after $e$ in $OPT$s list at time $t$ (i.e. is in $I_e^t$). Therefore, we have the following observation.
\begin{observation}
\label{delay.observation.number_of_good_elements_observatioon}
For each time $t$ and each element $e\in\mathbb{E}$ we have
$$|I_e^t|+|NI_e^t| = x_e^t-1 .$$
\end{observation}
The following lemma will come in handy later.
\begin{lemma}
\label{delay.lemma.help_for_2_different_events}
    At each time $t$ we have $$3x_e^t-4\cdot|I_e^t|-3x_e^t\cdot\mathbbm{1}{[x_e^t\leq4y_e^t]}\leq 0 .$$
\end{lemma}
\begin{proof}
If $x_e^t\leq4y_e^t$ then we have 
\begin{align*}
    3x_e^t-4\cdot|I_e^t|-3x_e^t\cdot\mathbbm{1}{[x_e^t\leq4y_e^t]}  = 3x_e^t-4\cdot|I_e^t|-3x_e^t\cdot1
     =  -4\cdot|I_e^t|
     \leq 0 .
\end{align*}
Henceforth we assume that $x_e^t>4y_e^t$, i.e. $y_e^t<\frac{1}{4}x_e^t$. We have $$|NI_e^t|\leq y_e^t-1<\frac{1}{4}x_e^t-1$$ where the left inequality is because $|NI_e^t|$ can be bounded by the number of elements which are before $e$ in $OPT$s list at time $t$, which is $y_e^t-1$. From Observation \ref{delay.observation.number_of_good_elements_observatioon} we have that
\begin{align*}
    |I_e^t|  = x_e^t-1-|NI_e^t|
     > x_e^t-1-(\frac{1}{4}x_e^t-1)
     = \frac{3}{4}x_e^t
\end{align*}
and hence
\begin{align*}
    3x_e^t-4\cdot|I_e^t|-3x_e^t\cdot\mathbbm{1}{[x_e^t\leq4y_e^t]} = 3x_e^t-4\cdot|I_e^t|
    < 3x_e^t-4\cdot\frac{3}{4}x_e^t = 0 .
\end{align*}
\end{proof}
The following observation and lemma will also come in handy later.
\begin{observation}
\label{delay.lemma.basic_request_counters_observation}
Let $N\in[n]$. Define
$$K=\bigcup_{\ell \in [N]} \{k\in [m]|e_k=loc_\ell^t \text{ and } r_k \text{ was either active or frozen with } RC_k \text{ in } ALG \text{ at time } t\} .$$
Then we have
$$\sum_{k\in K}RC_k^t\leq N .$$
\end{observation}
\begin{lemma}
\label{delay.lemma.request_counters_lemma}
Let $N\in[n]$. Define
$$K=\bigcup_{\ell\in[N]} \{k\in [m]|e_k=loc_\ell^t \text{ and } r_k \text{ was either active or frozen with } RC_k \text{ in } ALG \text{ at time } t\} .$$
Then we have
$$\sum_{k\in K}d_k\leq 2\cdot N .$$
\end{lemma}
\begin{proof}
We assume without loss of generality that $N=2^i$ for some $i\in\mathbb{N}\cup\{0\}$ (In the case where $N$ is not a power of $2$ the proof works too, although with some technical changes).
For each $0\leq j\leq i$ define
$$K_j=\bigcup_{2^{j-1} <  \ell \leq 2^{j}} \{k\in [m]|e_k=loc_\ell^t \text{ and } r_k \text{ was either active 
or frozen with } RC_k \text{ in } ALG \text{ at time } t\} .$$
In order to prove the lemma, we shall prove for each $0\leq j\leq i$ that $$\sum_{k\in K_j}d_k\leq 2^{j} .$$
If we do that, then the lemma will be proved since we have
$$\sum_{k\in K}d_k=\sum_{j=0}^{i}\sum_{k\in K_j}d_k\leq\sum_{j=0}^{i}2^{j}\leq 2^{i+1}=2\cdot N .$$
Let $0\leq j\leq i$. We consider 2 cases. The first case is that after time $t$, there is at least one prefix request counters event on $\ell$ that satisfies $2^{j-1}\leq \ell \leq n$ or an element counters event on an element $loc_\ell^t$ such that $2^{j-1}\leq\ell\leq n$. The second case is that such an event does not occur after time $t$. We begin with the first case. Let $t^{'}\geq t$ be the first time after $t$ when there is an event as mentioned above. Let $$K^{'}_j=\bigcup_{2^{j-1} <  \ell \leq 2^{j}} \{k\in [m]|e_k=loc_\ell^{t^{'}} \text{ and } r_k \text{ was either active 
or frozen with } RC_k \text{ in } ALG \text{ at time } {t^{'}}\} .$$
During the time between $t$ and $t^{'}$, it is guaranteed that the positions of the elements $loc_{2^{j-1}}^t,loc_{2^{j-1}+1}^t,...,loc_{n}^t$ in the list will stay the same. 
During this time, we also have that each request $r_k$ (such that $k\in K_j$) is guaranteed not to be deleted by $ALG$ (although it is possible that if it is active in $ALG$ at time $t$, it will become frozen with $RC_k$ before time $t^{'}$, but that does not cause any problem to the proof). Therefore, we have $K_j\subseteq K^{'}_j$. At time $t'$, during the prefix request counters event or element counter event mentioned above, for each $k\in K_j$, $r_k$ is going to become either frozen with $RC_k$ or frozen without $RC_k$ or deleted in $ALG$, thus it will not suffer any more delay penalty after the event at time $t^{'}$. Therefore we have $d_k(t^{'})=d_k$. By applying Observation \ref{delay.lemma.basic_request_counters_observation} on time $t^{'}$ we get that
$$\sum_{k\in K_j}d_k=\sum_{k\in K_j}d_k(t^{'})\leq \sum_{k\in K^{'}_j}d_k(t^{'}) \leq 2^j$$
where the first inequality is because $K_j\subseteq K^{'}_j$. Now we are left with the task to deal with the second case when after time $t$, there will not be any prefix request counters event on $\ell$ that satisfies $2^{j-1}\leq\ell\leq n$ and there will not be an element counters event on an element $loc_\ell^t$ such that $2^{j-1}\leq\ell\leq n$. The idea in dealing with this case to act as in the previous case where taking $t^{'}=\infty$. Define 
$$K^{\infty}_j=\bigcup_{2^{j-1} <  \ell \leq 2^{j}} \{k\in [m]|e_k=loc_\ell^{\infty} \text{ and } r_k \text{ was either active 
or frozen with } RC_k \text{ in } ALG \text{ at time } {\infty}\}$$
We have 
$$\sum_{k\in K_j}d_k=\sum_{k\in K_j}d_k(\infty)\leq \sum_{k\in K^{\infty}_j}d_k(t^{\infty}) \leq 2^j$$
where in the first inequality we used the fact that $K^j\subseteq K^{\infty}_j$ and in the second inequality we used Observation \ref{delay.lemma.basic_request_counters_observation} on time $\infty$. Observe that the fact that we are in the second case means that the positions of the elements $loc_{2^{j-1}}^t,loc_{2^{j-1}+1}^t,...,loc_{n}^t$ in $ALG$s list stay the same from time $t$ until time $\infty$.
\end{proof}
The following lemma bounds the cost $ALG$ pays and thus it will be used to change the model (i.e. rules) in which $ALG$ is being charged without increasing $ALG(\sigma)$ so anything we will prove later for the new model will also hold for the original model.
\begin{lemma}
\label{delay.lemma.cost_linear_with_delay}
For any time $t$, the cost of $ALG$ until time $t$ is bounded by
$$ALG^t(\sigma)\leq 6\cdot\sum_{k=1}^{m}d_k(t).$$
\end{lemma}
\begin{proof}
When a request $r_k$ suffers a delay penalty of $\epsilon$, the request counter $RC_k$ increases by $\epsilon$ and the element counter $EC_{e_k}$ increases by $\epsilon$. We deposit $\epsilon$ to the delay penalty, $2\epsilon$ to $RC_k$ and $3\epsilon$ to $EC_{e_k}$ (which sums up to $6\epsilon$). When there is a prefix request counters event on $\ell\in[n]$, then $ALG$ pays $2\ell$ for access, which is precisely the sum of the amount deposited on the request counters in the prefix $\ell$. When there is an element counter event on $e\in\mathbb{E}$ at position $\ell\in[n]$, $ALG$ pays precisely $2\ell$ access cost and $\ell-1$ swapping cost, i.e. at most $3\ell$ cost, which is what we deposited on $EC_e$.
\end{proof}
\begin{definition}
\label{delay.change_cost_model_alg}
We make the following change to the model: $ALG$ will not pay for accessing elements and swapping elements. However, whenever a request $r_k$ suffers a delay penalty of $\epsilon>0$, $ALG$ will be charged with a delay penalty of $6\epsilon$ instead of being charged with $\epsilon$.
\end{definition}
Definition \ref{delay.change_cost_model_alg}, changes $ALG(\sigma)$ to be $6\cdot\sum_{k=1}^{m}d_k$. Due to Lemma \ref{delay.lemma.cost_linear_with_delay} (when $t=\infty$), it is guaranteed that $ALG(\sigma)$ doesn't decrease as a result of this change.
\begin{observation}
\label{delay.lemma.change_cost_model_alg}
$ALG(\sigma)$ does not decrease as a result of the change of the model in Definition \ref{delay.change_cost_model_alg}
\end{observation}
\begin{definition}
We define the set of \textbf{events} ${P}$ which contains the following $7$ types of events:
\begin{enumerate}
\item A request which is active in both $ALG$ and $OPT$ suffers a delay penalty of $\epsilon$.
\item A request which is active in $ALG$ but not in $OPT$ (because $OPT$ has already served it) suffers a delay penalty of $\epsilon$.
\item A request which is active in $OPT$ but not in $ALG$ (because $ALG$ has already served it) suffers a delay penalty of $\epsilon$.
\item $ALG$ has a prefix request counters event.
\item $ALG$ has an element counter event.
\item $OPT$ serves multiple requests together.
\item $OPT$ swaps two elements.
\end{enumerate}
\end{definition}
Note that any event is a superposition of the above 7 events \footnote{Note that since $ALG$ is deterministic, it is the best interest of the adversary not to increase the delay penalty of a request $r_k$ after $ALG$ serves it because doing so may only cause $OPT(\sigma)$ to increase while $ALG(\sigma)$ will stay the same. Therefore, we could assume that the delay penalty of a request never increases after $ALG$ serves it, thus an event of type 3 never occurs and thus no needs to be analyzed. However, we do not use this assumption and allow the adversary to act in a non-optimal way for him. We take care of the case where events of type 3 may occur as well.}. Recall that the potential function $\Phi$ is defined in Section \ref{section_potential_function_delay} as follows:
\begin{align*}
    \Phi(t) & = \sum_{e\in\mathbb{E}}\rho_e(t)+36\cdot\sum_{k\in\lambda_1(t)}d_k(t)+\sum_{k\in\lambda_2(t)}(42d_k-6d_k(t))\cdot\mathbbm{1}{[x_k\leq4y_k]}+\\
    & +48\cdot\sum_{e\in\mathbb{E}}\mu_e(t)+8\cdot\sum_{k\in\lambda_2(t)}\frac{d_k-d_k(t)}{x_{e_k}^t}\cdot\mu_k(t)
\end{align*}
where the terms $\rho_e$, $d_k(t)$, $d_k$, $\lambda_1$, $\lambda_2$, $\mu_e$ and $\mu_k$ are also defined in that section.
\begin{definition}
For each event $p\in P$, we define:
\begin{itemize}
\item $ALG^p$ to be the cost $ALG$ pays during $p$.
\item $OPT^p$ to be the cost $OPT$ pays during $p$.
\item For any parameter $z$, $\Delta z^p$ to be the value of $z$ after $p$ minus the value of $z$ before $p$.
\end{itemize}
\end{definition}
We have $ALG(\sigma)\leq\sum_{p\in P}ALG^p$ \footnote{The reason why we have inequality instead of equality is that the change in Definition \ref{delay.change_cost_model_alg} may cause $ALG(\sigma)$ to increase.} and
$OPT(\sigma)=\sum_{p\in P}OPT^p$. Observe that $\Phi$ starts with $0$ (since at the beginning, the lists of $ALG$ and $OPT$ are identical) and is always non-negative. Therefore, if we prove that for each event $p\in P$, we have
$$ALG^p+\Delta\Phi^p\leq 336\cdot OPT^p$$
then, by summing it up for over all the events, we will be able to prove Theorem \ref{delay.thm.1}. Note that we do not care about the actual value $\Phi(t)$ by itself, for any time $t$. We will only measure the change of $\Phi$ as a result of each type of event in order to prove that the inequality mentioned above indeed holds. The 7 types of events that we will discuss are:
\begin{itemize}
\item The event when a request which is active in both $ALG$ and $OPT$ suffers a delay penalty of $\epsilon$ (event type 1) is analyzed in Lemma \ref{delay.lemma.event_type_1}.
\item The event when a request which is active in $ALG$ but not in $OPT$ (because $OPT$ has already served it) suffers a delay penalty of $\epsilon$ (event type 2) is analyzed in Lemma \ref{delay.lemma.event_type_2}.
\item The event when a request which is active in $OPT$ but not in $ALG$ (because $ALG$ has already served it) suffers a delay penalty of $\epsilon$ (event type 3) is analyzed in Lemma \ref{delay.lemma.event_type_3}.
\item The event when $ALG$ has a prefix request counters event (event type 4) is analyzed in Lemma \ref{delay.lemma.event_type_4}.
\item The event when $ALG$ has an element counter event (event type 5) is analyzed in Lemma \ref{delay.lemma.event_type_5}.
\item The event when $OPT$ serves multiple requests together (event type 6) is analyzed in Lemma \ref{delay.lemma.event_type_6}.
\item The event when $OPT$ swaps two elements (event type 7) is analyzed in Lemma \ref{delay.lemma.event_type_7}.
\end{itemize}
We now begin with taking care of event of type 1.
\begin{lemma}
\label{delay.lemma.event_type_1}
    Let $p\in P$ be the event where a request $r_k$ (where $k\in[m]$) which is active in both $ALG$ and $OPT$ suffers a delay penalty of $\epsilon$. We have
$$ALG^p+\Delta\Phi^p\leq 42\cdot OPT^p$$
\end{lemma}
\begin{proof}
Let $t$ be the time when $p$ begins. We have $ALG^p=6\epsilon$ (due to Definition \ref{delay.change_cost_model_alg}) and $OPT^p=\epsilon$, because $r_k$ is still active in $OPT$. For this reason we also have $k\in\lambda_1(t)$. Therefore, we are left with the task to prove that $$\Delta\Phi^p\leq 36\epsilon$$
We have $\Delta {d_k}^p=\Delta RC_k^p=\Delta EC_{e_k}^p=\epsilon$
and $\Delta\rho_{e_k}^p=-8\cdot|I_{e_k}^t|\cdot\frac{\epsilon}{x_{e_k}^t}\leq 0$. Apart from these changes, there are no changes that affect the potential function $\Phi$ as a result of the event $p$. Therefore the reader may observe that we indeed have
$
    \Delta\Phi^p \leq 36\epsilon
$.
\end{proof}
We now deal with event of type 2.
\begin{lemma}
\label{delay.lemma.event_type_2}
    Let $p\in P$ be the event where a request $r_k$ (where $k\in[m]$) which is active in $ALG$ but not in $OPT$ (because $OPT$ has already served it) suffers a delay penalty of $\epsilon$. We have
$$ALG^p+\Delta\Phi^p\leq 0(=OPT^p)$$
\end{lemma}
\begin{proof}
We have $OPT^p=0$, because $OPT$ has already served $r_k$. For this reason we also have $k\in\lambda_2(t)$. Let $t$ be the time when $p$ begins. We have $ALG^p=6\epsilon$ (due to Definition \ref{delay.change_cost_model_alg}). Therefore, we are left with the task to prove that $$6\epsilon+\Delta\Phi^p\leq 0$$
We have $\Delta {d_k}^p=\Delta RC_k^p=\Delta EC_{e_k}^p=\epsilon$ and $\Delta\rho_{e_k}^p=-8\cdot|I_{e_k}^t|\cdot\frac{\epsilon}{x_{e_k}^t}$ . Apart from these changes, there are no changes that affect the potential function $\Phi$ as a result of the event $p$. Therefore
$$
    \Delta\Phi^p = -8\cdot|I_{e_k}^t|\cdot\frac{\epsilon}{x_{e_k}^t}-6\epsilon\cdot\mathbbm{1}{[x_k\leq4y_k]}-8\cdot\frac{\epsilon}{x_{e_k}^t}\cdot\mu_k(t)
$$
Since $r_k$ is an active request in $ALG$, we have that $x_{e_k}^t=x_k$ (Lemma \ref{delay.lemma.xk_lemma}). Therefore, our task is to prove that 
$$
    6\epsilon-8\cdot|I_{e_k}^t|\cdot\frac{\epsilon}{x_{e_k}^t}-6\epsilon\cdot\mathbbm{1}{[x_{e_k}^t\leq4y_k]}-8\cdot\frac{\epsilon}{x_{e_k}^t}\cdot\mu_k(t)\leq 0
$$
Multiplying the above inequality by $\frac{x_{e_k}^t}{2\epsilon}>0$ yields that we need to prove the following:
$$
    3x_{e_k}^t-4\cdot|I_{e_k}^t|-3x_{e_k}^t\cdot\mathbbm{1}{[x_{e_k}^t\leq4y_k]}-4\cdot\mu_k(t)\leq 0
$$
i.e. we need to prove that
\begin{equation} \label{delay_event__epsilon_2_inequality_to_prove}
    3x_{e_k}^t-4\cdot(|I_{e_k}^t|+\mu_k(t))-3x_{e_k}^t\cdot\mathbbm{1}{[x_{e_k}^t\leq4y_k]}\leq 0
\end{equation}
For now, temporary assume that ever since $OPT$ served the request $r_k$ until time $t$, $OPT$ did not increase the position of the element $e_k$ in its list. Later we remove this assumption. The assumption means that $\mu_k(t)=0$. It also means that we must have $y_{e_k}^t\leq y_k$ (A strict inequality $y_{e_k}^t< y_k$ occurs in case $OPT$ performed at least one swap between $e_k$ and the previous element in its list after it served $r_k$ and before time $t$, thus decreasing the position of $e_k$ in its list). Therefore, we can replace $y_k$ with $y_{e_k}^t$ in Inequality \ref{delay_event__epsilon_2_inequality_to_prove} and conclude that it is sufficient for us to prove that $$3x_{e_k}^t-4\cdot|I_{e_k}^t|-3x_{e_k}^t\cdot\mathbbm{1}{[x_{e_k}^t\leq4y_{e_k}^t]}\leq 0$$
This was proven in Lemma \ref{delay.lemma.help_for_2_different_events}. Now we remove the assumption that ever since $OPT$ served $r_k$ until time $t$, $OPT$ did not increase the position of $e_k$ in its list. Our task is to prove that Inequality \ref{delay_event__epsilon_2_inequality_to_prove} continues to hold nonetheless. Assume that after $OPT$ served $r_k$ and before time $t$, $OPT$ performed a swap between $e_k$ and another element $e\in\mathbb{E}\backslash\{e_k\}$ where the position of $e_k$ in $OPT$s list was increased as a result of this swap (in other words, $e$ was the next element after $e_k$ in $OPT$s list before this swap). We should verify that Inequality \ref{delay_event__epsilon_2_inequality_to_prove} continues to hold nonetheless. The swap between $e$ and $e_k$ caused $|I_{e_k}^t|$ to either stay the same or decrease by $1$:
\begin{itemize}
    \item If $e$ was before $e_k$ in $ALG$s list at time $t$ then we had $e\in I_{e_k}^t$ before the swap between $e$ and $e_k$ in $OPT$s list and we will have $e\notin I_{e_k}^t$ after this swap (we will have $e\in NI_{e_k}^t$ after this swap).
    \item If $e_k$ was before $e$ in $ALG$s list at time $t$ then $I_{e_k}^t$ is not changed as a result of the swap between $e$ and $e_k$ in $OPT$s list.
\end{itemize}
Therefore, the swap between $e$ and $e_k$ in $OPT$s list caused $|I_{e_k}^t|$ to decrease by at most $1$. However, this swap certainly caused $\mu_k(t)$ to increase by $1$, which means it caused $|I_{e_k}^t|+\mu_k(t)$ to increase by at least $0$. Therefore, this swap could only decrease the left term of Inequality \ref{delay_event__epsilon_2_inequality_to_prove}, which means it will continue to hold nonetheless \footnote{Note that this swap doesn't affect the term $\mathbbm{1}{[x_k\leq4y_k]}$: $x_k$ is the position of $e_k$ in $ALG$s list so a swap in $OPT$ doesn't affect it. As for $y_k$, it is the position of $e_k$ in $OPT$s list when $OPT$ served $r_k$. Swaps which have been done after that (such as the swap we consider now) do not affect $y_k$.}. By applying the above argument for each swap $OPT$ done after it served $r_k$ until time $t$ which increased the position of $e_k$ in its list - we get that Inequality \ref{delay_event__epsilon_2_inequality_to_prove} will continue to hold nonetheless, thus the proof of the lemma is complete.
\end{proof}
Now that we dealt with event of type 2, we take care of event of type 3.
\begin{lemma}
\label{delay.lemma.event_type_3}
    Let $p\in P$ be the event where a request $r_k$ (where $k\in[m]$) which is active in $OPT$ but not in $ALG$ (because $ALG$ has already served it) suffers a delay penalty of $\epsilon$. We have
$$ALG^p+\Delta\Phi^p\leq OPT^p$$
\end{lemma}
\begin{proof}
We have $OPT^p=\epsilon$. We have $ALG^p=0$, because $ALG$ has already served $r_k$. For this reason we also have $k\notin\lambda_1(t),k\notin\lambda_2(t)$. Therefore, the reader can verify that $\Delta\Phi^p=0$ i.e. the potential function is not affected by the event $p$. Therefore we have $$ALG^p+\Delta\Phi^p=0+0<\epsilon= OPT^p$$
\end{proof}
Now we take care of event of type 4.
\begin{lemma} 
\label{delay.lemma.event_type_4}
Let $p\in P$ be the event where $ALG$ has a prefix request counters event. We have
$$ALG^p+\Delta\Phi^p= 0=OPT^p$$
\end{lemma}
\begin{proof}
Obviously, we have $OPT^p=0$. We also have $ALG^p=0$, due to Definition \ref{delay.change_cost_model_alg}, which implies that $ALG$ is charged only when a request suffers delay penalty and not during prefix request counters events and element counter events. Therefore, we are left with the task to prove that $\Delta\Phi^p=0$. Indeed, the only things which can happen as a result of this prefix request counters event is that active requests may become frozen and that requests $k\in[m]$ which are frozen with $RC_k$ may become frozen without $RC_k$. Neither of these changes have an effect on the potential function $\Phi$.
\end{proof}
Now we take care of event of type 5 - an element counter event in $ALG$. The following lemma will be used in the analysis of this event. After we take care of the event when $ALG$ has an element counters event, we will be left with the task to analyze the events which concern $OPT$ (events of types 6 and 7).
\begin{lemma}
\label{delay.lemma.bound_ro_in_element_counter_event}
    Let $p\in P$ be the event where $ALG$ has an element counter event on the element $e\in\mathbb{E}$ at time $t$. We have
$$\sum_{e^{'}\in\mathbb{E}}\Delta\rho_{e^{'}}^p\leq 36x_e^t-48\cdot|I_e^t|$$
\end{lemma}
\begin{proof}
The element counter $EC_e$ is zeroed in the event $p$. Its value before the event was equal to $x_e^t$ because this is the requirement for an element counter on $e$ to occur. Due to the move-to-front performed on $e$ by $ALG$, $x_e^t$ will be changed to $1$ and $|I_e^t|$ is changed to $0$ (We will have $I_e^t=\emptyset$ after the event $p$ because $e$ is going to be the first element in $ALG$s list, i.e. there will be no elements before it in $ALG$s list after the event). Therefore we have
\begin{align*}
    \Delta\rho_e^p =0\cdot(28-8\cdot\frac{0}{1})-|I_e^t|\cdot(28-8\cdot\frac{x_e^t}{x_e^t}) = -20\cdot |I_e^t|
\end{align*}
For each element $e^{'}\in\mathbb{E}$ that satisfies $x_{e^{'}}^t>x_e^t$it is easy to see that we have $\Delta\rho_{e^{'}}^p=0$.

For each element $e^{'}\in I_e^t$ we have that $EC_{e^{'}}^t$ doesn't change as a result of the event $p$ but $x_{e^{'}}^t$ will be increased by $1$ as a result of the move-to-front on $e$. $I_{e^{'}}^t$ will stay the same as a result of this move-to-front on $e$ in $ALG$s list because $e^{'}$ will still be before $e$ in $OPT$s list. Therefore we have
\begin{align*}
    \Delta\rho_{e^{'}}^p & = |I_{e^{'}}^t|\cdot(28-8\cdot\frac{EC_{e^{'}}^t}{x_{e^{'}}^t+1})-|I_{e^{'}}^t|\cdot(28-8\cdot\frac{EC_{e^{'}}^t}{x_{e^{'}}^t}) \\
    & = 8\cdot|I_{e^{'}}^t|\cdot EC_{e^{'}}^t\cdot(\frac{1}{x_{e^{'}}^t}-\frac{1}{x_{e^{'}}^t+1}) \\
    & = 8\cdot\frac{|I_{e^{'}}^t|\cdot EC_{e^{'}}}{x_{e^{'}}^t\cdot(x_{e^{'}}^t+1)}
    \leq 8
\end{align*}
where the inequality is due to:
\begin{itemize}
    \item $|I_{e^{'}}^t|\leq x_{e^{'}}^t-1$ because there are $x_{e^{'}}^t-1$ elements which are located before $e^{'}$ in $ALG$s list at time $t$ (before the move-to-front on $e$) and this gives us this bound to $|I_{e^{'}}^t|$.
    \item $EC_{e^{'}}\leq x_{e^{'}}^t$.
\end{itemize}
For each element $e^{'}\in NI_e^t$ we have that $EC_{e^{'}}^t$ doesn't change as a result of the event $p$ but $x_{e^{'}}^t$ will be increased by $1$ as a result of the move-to-front on $e$. $|I_{e^{'}}^t|$ 
will be increased by $1$ as a result of this move-to-front on $e$ for the following reason: Before the move-to-front on $e$, $e^{'}$ was located before $e$ in both $ALG$s list and $OPT$s list (and thus we had $e^{'}\notin I_{e^{'}}^t$) but after the move-to-front on $e$ in $ALG$s list, $e$ will be before $e^{'}$ in $ALG$s list (and thus we will have $e^{'}\in I_{e^{'}}^t$). Therefore we have
\begin{align*}
    \Delta\rho_{e^{'}}^p & = (|I_{e^{'}}^t|+1)\cdot(28-8\cdot\frac{EC_{e^{'}}^t}{x_{e^{'}}^t+1})-|I_{e^{'}}^t|\cdot(28-8\cdot\frac{EC_{e^{'}}^t}{x_{e^{'}}^t}) \\
    & = 8\cdot|I_{e^{'}}^t|\cdot(\frac{1}{x_{e^{'}}^t}-\frac{1}{x_{e^{'}}^t+1})+28-8\cdot\frac{EC_{e^{'}}^t}{x_{e^{'}}^t+1} \\
    & \leq 8\cdot|I_{e^{'}}^t|\cdot(\frac{1}{x_{e^{'}}^t}-\frac{1}{x_{e^{'}}^t+1})+28\\
    & = 8\cdot\frac{|I_{e^{'}}^t|\cdot EC_{e^{'}}}{x_{e^{'}}^t\cdot(x_{e^{'}}^t+1)}+28\\
    & \leq 8+28
     = 36
\end{align*}
where the second inequality is that we have again $|I_{e^{'}}^t|\leq x_{e^{'}}^t-1$ and $EC_{e^{'}}\leq x_{e^{'}}^t$.
Therefore we have 
\begin{align*}
\sum_{e^{'}\in\mathbb{E}}\Delta\rho_{e^{'}}^p & = \sum_{e^{'}\in\mathbb{E}:x_{{e^{'}}^t}<x_e^t}\Delta\rho_{e^{'}}^p+\Delta\rho_e^p+\sum_{e^{'}\in\mathbb{E}:x_{{e^{'}}^t}>x_e^t}\Delta\rho_{e^{'}}^p \\
& = \sum_{e^{'}\in I_e^t}\Delta\rho_{e^{'}}^p+\sum_{e^{'}\in NI_e^t}\Delta\rho_{e^{'}}^p+\Delta\rho_e^p+\sum_{e^{'}\in\mathbb{E}:x_{{e^{'}}^t}>x_e^t}\Delta\rho_{e^{'}}^p\\
& \leq \sum_{e^{'}\in I_e^t}8+\sum_{e^{'}\in NI_e^t}36-20\cdot |I_e^t|+\sum_{e^{'}\in\mathbb{E}:x_{{e^{'}}^t}>x_e^t}0 \\
& = 8\cdot|I_e^t|+36\cdot|NI_e^t|-20\cdot |I_e^t|\\
& = 36\cdot|NI_e^t|-12\cdot |I_e^t|\\
& = 36\cdot(x_e^t-1-|I_e^t|)-12\cdot |I_e^t|\\
& = 36x_e^t-48\cdot|I_e^t|-36
 \leq 36x_e^t-48\cdot|I_e^t|
\end{align*}
\end{proof}
Now we are ready to take care of event of type 5.
\begin{lemma}
\label{delay.lemma.event_type_5}
    Let $p\in P$ be the event where $ALG$ has an element counter event on the element $e\in\mathbb{E}$ at time $t$. We have
$$ALG^p+\Delta\Phi^p\leq 0(=OPT^p)$$
\end{lemma}
\begin{proof}
Obviously, we have $OPT^p=0$. We also have $ALG^p=0$, due to Definition \ref{delay.change_cost_model_alg}, which implies that $ALG$ is charged only when a request suffers delay penalty and not during prefix request counters events and element counter events. Therefore, we are left with the task to prove that $$\Delta\Phi^p\leq0$$
We define the following:
$$K_1=\{k\in\lambda_1(t)|e_k=e\}$$
$$K_2=\{k\in\lambda_2(t)|e_k=e\}$$
$$K=\{k\in\lambda(t)|e_k=e\}=K_1\cup K_2$$
We have $\Delta\mu_e=-\mu_e(t)$ because $\mu_e$ is set to $0$ after this element counter event. For each $k\in K$, the request $r_k$ is going to be deleted in $ALG$ after the event $p$. This yields to 2 consequences:
\begin{itemize}
    \item $d_k(t)=d_k$.
    \item We will have $k\notin \lambda_1(t)$,$k\notin \lambda_2(t)$ after the event $p$.
\end{itemize}
Therefore
\begin{align*}
\Delta\Phi^p & = \sum_{e\in\mathbb{E}}\Delta\rho_e^p-36\cdot\sum_{k\in K_1}d_k(t)-\sum_{k\in K_2}(42d_k-6d_k(t))\cdot\mathbbm{1}{[x_k\leq4y_k]}\\
    & -48\cdot\mu_e(t)-8\cdot\sum_{k\in K_2}\frac{d_k-d_k(t)}{x_e^t}\cdot\mu_k(t)\\
    & = \sum_{e\in\mathbb{E}}\Delta\rho_e^p-36\cdot\sum_{k\in K_1}d_k-\sum_{k\in K_2}(42d_k-6d_k)\cdot\mathbbm{1}{[x_k\leq4y_k]}\\
    & -48\cdot\mu_e(t)-8\cdot\sum_{k\in K_2}\frac{d_k-d_k}{x_e^t}\cdot\mu_k(t)\\
    & = \sum_{e\in\mathbb{E}}\Delta\rho_e^p-36\cdot\sum_{k\in K_1}d_k-36\cdot\sum_{k\in K_2}d_k\cdot\mathbbm{1}{[x_k\leq4y_k]}-48\cdot\mu_e(t)\\
    & \leq 36x_e^t-48\cdot|I_e^t|-36\cdot\sum_{k\in K_1}d_k-36\cdot\sum_{k\in K_2}d_k\cdot\mathbbm{1}{[x_k\leq4y_k]}-48\cdot\mu_e(t)\\
    & = 36x_e^t-48\cdot(|I_e^t|+\mu_e(t))-36\cdot\sum_{k\in K_1}d_k-36\cdot\sum_{k\in K_2}d_k\cdot\mathbbm{1}{[x_k\leq4y_k]}\\
    & \leq 36x_e^t-48\cdot(|I_e^t|+\mu_e(t))-36\cdot\sum_{k\in K_1}d_k-36\cdot\sum_{k\in K_2}d_k\cdot\mathbbm{1}{[x_e^t\leq4y_k]}
\end{align*}
where in the second equality we used that $d_k(t)=d_k$ for each $k\in K$, in the first inequality we used Lemma \ref{delay.lemma.bound_ro_in_element_counter_event} and in the second inequality we used that $x_k\leq x_e^t$, which follows from Lemma \ref{delay.lemma.xk_lemma}. This allowed us to replace $x_k$ with $x_e^t$. Therefore we are left with the task to prove that
$$36x_e^t-48\cdot(|I_e^t|+\mu_e(t))-36\cdot\sum_{k\in K_1}d_k-36\cdot\sum_{k\in K_2}d_k\cdot\mathbbm{1}{[x_e^t\leq4y_k]}\leq 0$$
i.e. (after dividing the above inequality by $12$)
\begin{equation} \label{delay_element_counter_event_inequality_to_prove}
3x_e^t-4\cdot(|I_e^t|+\mu_e(t))-3\cdot\sum_{k\in K_1}d_k-3\cdot\sum_{k\in K_2}d_k\cdot\mathbbm{1}{[x_e^t\leq4y_k]}\leq 0
\end{equation}
For now, temporary assume that ever since the last element counter event on $e$ (and if $p$ is the first element counter event on $e$ - then ever since the beginning, i.e. time $0$) until time $t$, $OPT$ did not increase the position of the element $e$ in its list. Later we remove this assumption. The assumption means that $\mu_e(t)=0$. Plugging this in Inequality \ref{delay_element_counter_event_inequality_to_prove} yields that our task is to prove that 
$$3x_e^t-4\cdot|I_e^t|-3\cdot\sum_{k\in K_1}d_k-3\cdot\sum_{k\in K_2}d_k\cdot\mathbbm{1}{[x_e^t\leq4y_k]}\leq 0$$
The assumption also means that for each $k\in K_2$ we have $y_e^t\leq y_k$ (A strict inequality $y_e^t< y_k$ occurs in case $OPT$ performed at least one swap between $e$ and the previous element in its list after it served $r_k$ and before time $t$, thus decreasing the position of $e$ in its list). Therefore (by replacing $y_k$ with $y_e^t$ in the above inequality), it is sufficient to prove that $$3x_e^t-4\cdot|I_e^t|-3\cdot\sum_{k\in K_1}d_k-3\cdot\sum_{k\in K_2}d_k\cdot\mathbbm{1}{[x_e^t\leq4y_e^t]}\leq 0$$
We analyze more the left term of the above inequality. Observe that
\begin{align*}
    & 3x_e^t-4\cdot|I_e^t|-3\cdot\sum_{k\in K_1}d_k-3\cdot\sum_{k\in K_2}d_k\cdot\mathbbm{1}{[x_e^t\leq4y_e^t]}\\
    \leq & \text{ }3x_e^t-4\cdot|I_e^t|-3\cdot\sum_{k\in K_1}d_k\cdot\mathbbm{1}{[x_e^t\leq4y_e^t]}-3\cdot\sum_{k\in K_2}d_k\cdot\mathbbm{1}{[x_e^t\leq4y_e^t]}\\
    = & \text{ }3x_e^t-4\cdot|I_e^t|-3\cdot\sum_{k\in K}d_k\cdot\mathbbm{1}{[x_e^t\leq4y_e^t]}\\
    = & \text{ }3x_e^t-4\cdot|I_e^t|-3x_e^t\cdot\mathbbm{1}{[x_e^t\leq4y_e^t]}
\end{align*}
where the first equality is because $K=K_1\cup K_2$ and the second equality is because the event $p$ is an element counter event on the element $e$ at time $t$, thus we have
$$\sum_{k\in K}d_k=x_e^t$$
From the analysis above we get that it is sufficient to prove that
$$3x_e^t-4\cdot|I_e^t|-3x_e^t\cdot\mathbbm{1}{[x_e^t\leq4y_e^t]}\leq 0$$
This was proven in Lemma \ref{delay.lemma.help_for_2_different_events}. Now we remove the assumption that ever since the last element counter event on $e$ until time $t$, $OPT$ did not increase the position of $e$ in its list. Our task is to prove that Inequality \ref{delay_element_counter_event_inequality_to_prove} continues to hold nonetheless. Assume that after the last element counter event on $e$ (or, if $p$ is the first element counter event on $e$, assume that after time $0$) and before time $t$, $OPT$ performed a swap between $e$ and another element $e^{'}\in\mathbb{E}\backslash\{e\}$ where the position of $e$ in $OPT$s list was increased as a result of this swap (in other words, $e^{'}$ was the next element after $e$ in $OPT$s list before this swap). We should verify that Inequality \ref{delay_element_counter_event_inequality_to_prove} continues to hold nonetheless. Firstly, note that for each $k\in K$, this swap between $e$ and $e^{'}$ may only increase $y_k^t$, which may only decrease the term $-3\cdot\sum_{k\in K}d_k\cdot\mathbbm{1}{[x_e^t\leq4y_k]}$. The swap between $e$ and $e^{'}$ caused $|I_e^t|$ to either stay the same or decrease by $1$:
\begin{itemize}
    \item If $e^{'}$ was before $e$ in $ALG$s list at time $t$ then we had $e^{'}\in I_{e}^t$ before the swap between $e$ and $e^{'}$ in $OPT$s list and we will have $e^{'}\notin I_{e}^t$ after this swap (we will have $e^{'}\in NI_{e}^t$ after this swap).
    \item If $e$ was before $e^{'}$ in $ALG$s list at time $t$ then $I_{e}^t$ is not changed as a result of the swap between $e$ and $e^{'}$ in $OPT$s list.
\end{itemize}
Therefore, the swap between $e^{'}$ and $e$ in $OPT$s list caused $|I_{e}^t|$ to decrease by at most $1$. However, this swap certainly caused $\mu_e(t)$ to increase by $1$, which means it caused $|I_{e}^t|+\mu_e(t)$ to increase by at least $0$. Therefore, this swap could only decrease the left term of Inequality \ref{delay_element_counter_event_inequality_to_prove}, which means it will continue to hold nonetheless. By applying the above argument for each swap $OPT$ done after it served $r_k$ until time $t$ which increased the position of $e$ in its list - we get that Inequality \ref{delay_element_counter_event_inequality_to_prove} will continue to hold nonetheless, thus the proof of the lemma is complete.
\end{proof}
Now that we analyzed the event of type 5, we analyzed all the events which concern $ALG$. We are left with the task to analyze the events which concern $OPT$ - event types 6 and 7. We now begin with event of type 6.
\begin{lemma}
\label{delay.lemma.event_type_6}
    Let $p\in P$ be the event where $OPT$ serves multiple requests together at time $t$. Then
$$ALG^p+\Delta\Phi^p\leq 336\cdot OPT^p$$
\end{lemma}
\begin{proof}
Obviously, we have $ALG^p=0$. Denote by $y\in[n]$ to be the position of the farthest element that $OPT$ accessed in the event $p$,. Then $OPT^p=y$. Therefore, our task is to prove that 
$$\Delta\Phi^p\leq 336y$$
Recall that both $ALG$s list and $OPT$s list are not changed during the event $p$. During the event $p$, $OPT$ accessed the first $y$ elements in its list, thus serving all the requests which were active in $OPT$ for these elements prior to the access operation performed by $OPT$ in the event $p$. 

Let $K\subseteq[m]$ be the set of requests (request indices) that $OPT$ served during the event $p$, i.e. a request index $k\in[m]$ satisfies $k\in K$ if it fulfills \textbf{both} the following 2 requirements:
\begin{itemize}
    \item $r_k$ was active in $OPT$ prior to the event $p$.
    \item $e_k$ was one of the first $y$ elements in $OPT$s list at time $t$.
\end{itemize}
For each $k\in K$ we must have $1\leq y_{e_k}^t\leq y$. We also define:
\begin{itemize}
    \item $K^{act}=\{k\in K|r_k\text{ was active in } ALG \text{ during the event } p\}$.
    \item $K^{frz}=\{k\in K|r_k\text{ was frozen in } ALG \text{ during the event } p\}$.
    \item $K^{del}=\{k\in K|r_k\text{ was deleted in } ALG \text{ during the event } p\}$.
\end{itemize}
Observe that $K=K^{act}\cup K^{frz}\cup K^{del}$. Note that if $k\in K^{del}$ then $r_k$ has already been deleted by $ALG$ before the event $p$ and therefore $OPT$ serving $r_k$ during the event $p$ does not affect the potential function $\Phi$ at all (we already had $k\notin\lambda_1(t)$ and $k\notin\lambda_2(t)$). Henceforth we will consider only the request indices $K^{act}\cup K^{frz}$.
Due to the definitions above and the definitions of $\lambda_1(t)$ and $\lambda_2(t)$, we have for each $k\in K^{act}\cup K^{frz}$ that $k$ moves from $\lambda_1(t)$ to $\lambda_2(t)$. Hence, the change in the potential function $\Phi$ is
$$-36\cdot\sum_{k\in K^{act}\cup K^{frz}}d_k(t)+\sum_{k\in K^{act}\cup K^{frz}}(42d_k-6d_k(t))\cdot\mathbbm{1}{[x_k\leq4y_k]}$$
Observe that the potential function $\Phi$ does not suffer any more changes as a result of the event $p$ apart from the changes discussed above.
We define
$$M=\bigcup_{\ell \in [\min\{4y,n\}]} \{k\in K^{act}|e_k=loc_\ell^t\}$$
We have
\begin{align*}
\Delta\Phi^p & = -36\cdot\sum_{k\in K^{act}\cup K^{frz}}d_k(t)+\sum_{k\in K^{act}\cup K^{frz}}(42d_k-6d_k(t))\cdot\mathbbm{1}{[x_k\leq4y_k]} \\
& \stackrel{(1)}{\leq} 42\cdot\sum_{k\in K^{act}}d_k\cdot\mathbbm{1}{[x_{e_k}^t\leq4y_{e_k}^t]}+\sum_{k\in K^{frz}}-36\cdot d_k(t)+(42d_k-6d_k(t))\cdot\mathbbm{1}{[x_k\leq4y_k]}\\
& \stackrel{(2)}{=} 42\cdot\sum_{k\in K^{act}}d_k\cdot\mathbbm{1}{[x_{e_k}^t\leq4y_{e_k}^t]}+\sum_{k\in K^{frz}}-\underbrace{36\cdot d_k+(42d_k-6d_k)\cdot\mathbbm{1}{[x_k\leq4y_k]}}_{\leq 0}\\
& \leq 42\cdot\sum_{k\in K^{act}}d_k\cdot\mathbbm{1}{[x_k\leq4y_k]}\\
& \stackrel{(3)}{=} 42\cdot\sum_{k\in K^{act}}d_k\cdot\mathbbm{1}{[x_{e_k}^t\leq4y_{e_k}^t]}\\
& = 42\cdot\sum_{\ell=1}^{n}\sum_{k\in K^{act}:e_k=loc_\ell^t}d_k\cdot\mathbbm{1}{[x_{e_k}^t\leq4y_{e_k}^t]}\\
& \stackrel{(4)}{=} 42\cdot\sum_{\ell=1}^{n}\sum_{k\in K^{act}:e_k=loc_\ell^t}d_k\cdot\mathbbm{1}{[x_{loc_\ell^t}^t\leq4y_{loc_\ell^t}^t]}\\
& = 42\cdot(\sum_{\ell=1}^{\min\{4y,n\}}\sum_{k\in K^{act}:e_k=loc_\ell^t}d_k\cdot\mathbbm{1}{[x_{loc_\ell^t}^t\leq4y_{loc_\ell^t}^t]}+\sum_{\ell=\min\{4y,n\}+1}^{n}\sum_{k\in K^{act}:e_k=loc_\ell^t}d_k\cdot\mathbbm{1}{[x_{loc_\ell^t}^t\leq4y_{loc_\ell^t}^t]})\\
& \stackrel{(5)}{\leq} 42\cdot\sum_{\ell=1}^{\min\{4y,n\}}\sum_{k\in K^{act}:e_k=loc_\ell^t}d_k\\
& = 42\cdot\sum_{k\in M}d_k\\
& \stackrel{(6)}{\leq} 42\cdot 2\cdot \min\{4y,n\}
 \leq 42\cdot 2\cdot 4y
 = 336y
\end{align*}
Inequality (1) holds because for each $k\in K^{act}$ we have $d_k(t)\geq 0$. Equality (2) holds because for each $k\in K^{frz}$ we have $d_k(t)=d_k$ (Observation \ref{delay.lemma.dk_frozen}).
Equality (3) is because $y_k=y_{e_k}^t$ for each $k\in K^{act}$: it follows straight from the definitions. In Equality (4) we just plugged in $e_k=loc_\ell^t$. Inequality (6) is due to Lemma \ref{delay.lemma.request_counters_lemma}. We are left with the task of explaining Inequality (5). It is obvious that $\mathbbm{1}{[x_{loc_\ell^t}^t\leq4y_{loc_\ell^t}^t]}\leq 1$ for each $1\leq \ell \leq \min\{4y,n\}$. Therefore, in order to justify Inequality (5) (and thus complete the proof of the lemma), we should prove that $$\sum_{\ell=\min\{4y,n\}+1}^{n}\sum_{k\in K^{act}:e_k=loc_\ell^t}d_k\cdot\mathbbm{1}{[x_{loc_\ell^t}^t\leq4y_{loc_\ell^t}^t]}=0$$
If $n\leq 4y$ then $\min\{4y,n\}+1=n+1>n$ so it is obvious because we do not sum anything. Otherwise, we have $4y<n$ and thus
\begin{align*}
    0 & \leq \sum_{\ell=\min\{4y,n\}+1}^{n}\sum_{k\in K^{act}:e_k=loc_\ell^t}d_k\cdot\mathbbm{1}{[x_{loc_\ell^t}^t\leq4y_{loc_\ell^t}^t]} \\
    & = \sum_{\ell=4y+1}^{n}\sum_{k\in K^{act}:e_k=loc_\ell^t}d_k\cdot\mathbbm{1}{[x_{loc_\ell^t}^t\leq4y_{loc_\ell^t}^t]} \\
    & = \sum_{\ell=4y+1}^{n}\sum_{k\in K^{act}:e_k=loc_\ell^t}d_k\cdot\mathbbm{1}{[\ell\leq4y_{loc_\ell^t}^t]}\\
    & \leq \sum_{\ell=4y+1}^{n}\sum_{k\in K^{act}:e_k=loc_\ell^t}d_k\cdot\mathbbm{1}{[\ell\leq4y]}\\
    & = \sum_{\ell=4y+1}^{n}\sum_{k\in K^{act}:e_k=loc_\ell^t}d_k\cdot0
     = 0
\end{align*}
where the first equality is because $\min\{4y,n\}=4y$, the second inequality is because $x_{loc_\ell^t}^t=\ell$ and the second inequality is because for each $k\in K^{act}$ we must have $y_{loc_\ell^t}^t\leq y$.
\end{proof}
Now we are left with the task to analyze the last event - the event of type 7. The following lemma will be needed to deal with this event.
\begin{lemma}
\label{delay.lemma.bound_ro_opt_swaps_two_elements}
    Let $p\in P$ be the event where $OPT$ swaps two elements $i,j\in\mathbb{E}$ at time $t$ where $j$ was the next element after $i$ in $OPT$s list prior to this swap. Then
$$\sum_{e\in\mathbb{E}}\Delta\rho_e^p\leq 28$$
\end{lemma}
\begin{proof}
We have $\Delta\rho_e^p=0$. Therefore we have $$\sum_{e\in\mathbb{E}}\Delta\rho_e^p=\Delta\rho_i^p+\Delta\rho_j^p+\sum_{e\in\mathbb{E}\setminus\{i,j\}}\Delta\rho_e^p=\Delta\rho_i^p+\Delta\rho_j^p+\sum_{e\in\mathbb{E}\setminus\{i,j\}}0=\Delta\rho_i^p+\Delta\rho_j^p$$
Therefore, our task is to prove that 
$$\Delta\rho_i^p+\Delta\rho_j^p\leq 28$$
We consider two cases:
\begin{itemize}
    \item The case where $i$ is located before $j$ in $ALG$s list at time $t$. We had $i\in NI_j^t$ before the event $p$ and after the event $p$ we will have $i\in I_j^t$. In other words, $|I_j^t|$ increases by $1$ as a result of the event $p$. Therefore, $$\Delta\rho_j^p=1\cdot(28-8\cdot\frac{EC_j^t}{x_j^t})\leq 28$$ Since $i$ is located before $j$ in $ALG$s list, the swap performed by $OPT$ between $i$ and $j$ doesn't change $I_i^t$ and therefore we have $\Delta\rho_i^p=0$. Therefore we have $$\Delta\rho_i^p+\Delta\rho_j^p\leq 0+28=28$$
    \item The case where $i$ is located after $j$ in $ALG$s list at time $t$. We had $j\in I_i^t$ before the event $p$ and after the event $p$ we will have $j\in NI_i^t$ (and $j\notin I_i^t$). In other words, $|I_i^t|$ decreases by $1$ as a result of the event $p$. Therefore, $$\Delta\rho_i^p=-1\cdot(28-8\cdot\frac{EC_i^t}{x_i^t})=-28+8\cdot\frac{EC_i^t}{x_i^t}\leq -28+8\cdot\frac{x_i^t}{x_i^t}=-20$$ where in the inequality we used the fact that we always have $0\leq EC_i^t\leq x_i^t$. Since $j$ is located before $i$ in $ALG$s list, the swap performed by $OPT$ between $i$ and $j$ doesn't change $I_j^t$ and therefore we have $\Delta\rho_j^p=0$. Therefore we have $$\Delta\rho_i^p+\Delta\rho_j^p\leq -20+0=-20$$
\end{itemize}
In either case, we have $\Delta\rho_i^p+\Delta\rho_j^p\leq 28$ and thus the lemma has been proven.
\end{proof}
Now we are ready to deal with the event of type 7.
\begin{lemma}
\label{delay.lemma.event_type_7}
    Let $p\in P$ be the event where $OPT$ swaps two elements $i,j\in\mathbb{E}$ at time $t$ where $j$ was the next element after $i$ in $OPT$s list prior to this swap. Then
$$ALG^p+\Delta\Phi^p\leq 84\cdot OPT^p$$
\end{lemma}
\begin{proof}
We have $ALG^p=0$ and $OPT^p=1$, thus our target is to prove that $$\Delta\Phi^p\leq 84$$
We have $\Delta\mu_i^p=1$ and $\Delta\mu_j^p=0$. Therefore we have $$\sum_{e\in\mathbb{E}}\Delta\mu_e^p=\Delta\mu_i^p+\sum_{e\in\mathbb{E}\setminus\{i\}}\Delta\mu_e^p=1+\sum_{e\in\mathbb{E}\setminus\{i\}}0=1$$
We define:
$$K_2^{act}=\{k\in\lambda_2(t)|e_k=i \text{ and } r_k \text{ was active in } ALG \text{ at time } t\}$$
$$K_2^{frz}=\{k\in\lambda_2(t)|e_k=i \text{ and } r_k \text{ was frozen (with or without } RC_k\text{) in } ALG \text{ at time } t\}$$
$$K_2=\{k\in\lambda_2(t)|e_k=i\}=K_2^{act}\cup K_2^{frz}$$
$K_2^{act}$ is the set of all the requests (request indices) for the element $i$ which have been served by $OPT$ before time $t$ and were active in $ALG$ at time $t$. $K_2^{frz}$ is the set of all the requests (request indices) for the element $i$ which have been served by $OPT$ before time $t$ and were frozen in $ALG$ at time $t$. We have $\Delta\mu_k^p=1$ for each $k\in K_2$ and $\Delta\mu_k^p=0$ for each $k\in [m]\setminus K_2$. Therefore we have

\begin{align*}
    \Delta\Phi^p & = \sum_{e\in\mathbb{E}}\Delta\rho_e^p+48\cdot\underbrace{\sum_{e\in\mathbb{E}}\Delta\mu_e^p}_{=1}+8\cdot\sum_{k\in K_2}\frac{d_k-d_k(t)}{x_{e_k}^t}\cdot\underbrace{\Delta\mu_k^p}_{=1} \\ 
    & \leq 28+48\cdot1+8\cdot\sum_{k\in K_2^{act}}\frac{d_k-d_k(t)}{x_{e_k}^t}+8\cdot\sum_{k\in K_2^{frz}}\frac{d_k-d_k(t)}{x_{e_k}^t} \\
    & \leq 76+8\cdot\sum_{k\in K_2^{act}}\frac{d_k-0}{x_{e_k}^t}+8\cdot\sum_{k\in K_2^{frz}}\frac{d_k-d_k}{x_{e_k}^t} \\
    & = 76+8\cdot\frac{1}{x_{e_k}^t}\cdot\sum_{k\in K_2^{act}}d_k \\
    & \leq 76+8\cdot\frac{1}{x_{e_k}^t}\cdot x_{e_k}^t
     = 84
\end{align*}
where the first inequality is due to Lemma \ref{delay.lemma.bound_ro_opt_swaps_two_elements}. The second inequality is because we always have $d_k(t)\geq 0$ and if $r_k$ is frozen in $ALG$ at time $t$ then we can use Observation \ref{delay.lemma.dk_frozen} and get that $d_k(t)= d_k$. The third inequality is due to Observation \ref{delay.lemma.basic_request_counters_observation}.
\end{proof}
Now that we analyzed all the possible events, we are ready to prove Theorem \ref{delay.thm.1}.
\begin{proof}[Proof of Theorem \ref{delay.thm.1}]
In the previous lemmas we have proven for each event $p\in P$ that
$$ALG^p+\Delta\Phi^p\leq 336\cdot OPT^p$$
The theorem follows by summing up for over all events and use the fact that $\Phi$ starts with $0$ and is always non-negative.
\end{proof}
\section{Failed Algorithms for Delay}
\label{section_failed_algorithms}
In this section we present some intuitive algorithms for List Update with Delay which are simpler than Algorithm \ref{delay.alg} but unfortunately fail to achieve a constant competitive ratio. 
Note that in the time windows version, we have presented the $O(1)$-competitive Algorithm \ref{deadlines.alg}, which is (in our opinion) the most intuitive and simple one for the problem.
In contrast 
Algorithm \ref{delay.alg} for the delay version may not seem to be the most simple algorithm. The goal of this section is to justify the complexity of Algorithm \ref{delay.alg} in order to get a constant competitive ratio. Specifically, the algorithm uses two types of events: prefix request counters event and element counter event. We would like to show that using only one type of events does not work. Obviously, there are other possible algorithms for the problem which are not discussed here. 


For each algorithm $A$ that we present in this section, we present a request sequence $\sigma$ such that $\frac{A(\sigma)}{OPT(\sigma)}\geq \omega(1)$. The sequence $\sigma$ is a sequence for the price collecting problem (which is a special case of the delay and hence it is also a counter example to the delay version). In the price collecting problem, the request $r_k$ is represented by a tuple: $r_k=(e_k,a_k,q_k,p_k)$ where $e_k\in\mathbb{E}$ is the required element, $a_k$ is the arrival time of the request, $q_k$ is the deadline and $p_k$ is the penalty price. 
For each $r_k$ an algorithm has a choice: either access the element $e_k$ between $a_k$ and $q_k$ or pay the penalty price of $p_k$ for not doing so. 
We denote the initial list of the algorithm (as well as $OPT$) as $c_1,c_2,...,c_n$. Whenever used, $0<\epsilon<1$ is a very small number. We assume without loss of generality that $n$ is an integer square (otherwise round it down to a square).

Since the algorithms that we discuss here use one type of event, we need to consider one type of counters. All the algorithms we present in this section use element counters, as in Algorithm \ref{delay.alg}: for each $\ell\in[n]$, the element counter $EC_{c_\ell}$ is initialized with $0$ at time $t=0$ and suffers delay penalty as active requests for the element $c_\ell$ suffer delay penalty. Alternatively we could have used request counters and view the element counter $EC_{c_\ell}$ as the sum of the the request counters for $c_\ell$ which have not been deleted yet.

The first algorithm we present acts upon element counters events. Recall that an element counter event on $e\in\mathbb{E}$ occurs when $EC_e$ reaches $e$'s current position in the lits.
The other algorithms act upon prefix element counter events. 
A prefix element counter event on $\ell\in[n]$ occurs when the sum of all the element counters of the first $\ell$ elements in the list reaches the value of $\ell$.


Now we show the algorithms and their counter examples.

\textbf{Algorithm $A_1$}: Upon element counter event on $e\in\mathbb{E}$: Serve the set of requests in the first $2{EC}_e$ elements in the list, set ${EC}_e$ to $0$ and Move-to-front the element $e$ \footnote{If we change Algorithm \ref{delay.alg} from Section \ref{section_algorithm_delay} so it will ignore prefix request counters events and act only upon element counter events then the request counters become useless and thus we get an algorithm which is equivalent to Algorithm $A_1$.}.

\textbf{Counter example for Algorithm $A_1$}: Define $m=n$, for each $\ell\in[n]$ define $r_\ell=(c_\ell,0,0,\ell-\epsilon)$. Algorithm $A_1$ does not do anything (no element counter event occurs) and therefore its cost is the total price for the requests i.e. $A_1(\sigma)=\sum_{\ell=1}^n (\ell-\epsilon)=\Theta(n^2)$. $OPT$ serves all the $n$ requests together at time $t=0$ (by accessing the entire list with an access cost of $n$) and does not pay any penalty price. Both $A_1$ and $OPT$ do not perform any swap and thus
$\frac{A_1(\sigma)}{OPT(\sigma)}=\frac{\Theta(n^2)}{n}=\Omega(n)$.

\textbf{Algorithm $A_2$}: Upon prefix element counter event on $\ell\in[n]$: Let $e$ be the $\ell$-th element in the list currently. Serve the set of requests in the first $2\ell$ elements in the list, set ${EC}_e$ to $0$ and Move-to-front the element $e$ \footnote{Observe that the value of $EC_e$ at the beginning of the event must be bigger than $0$, otherwise we would not be in a prefix element counter event on $\ell$.}.

\textbf{Counter example for Algorithm $A_2$}: Define $m=2n-1$, $r_1=(c_1,0,2n,1-\epsilon)$ and for each $\ell\in[n-1]$ define $r_{2\ell}=(c_{\ell+1},2\ell,2\ell,1+\epsilon)$, $r_{2\ell+1}=(c_1,2\ell+1,2n,1)$. For each $\ell\in[n-1]$, $A_2$ has a prefix element counter event on $\ell+1$ at time $2\ell$: At the beginning of the event, the element $c_{\ell+1}$ is located at position $\ell+1$ and the element $c_1$ is located at position $\ell$. The event is caused due to the request $r_{2\ell}$ and the previous requests for the element $c_1$. During this event, $A_2$ pays an access cost of $\Theta(\ell)$ and a swapping cost of $\Theta(\ell)$ (for moving the element $c_{\ell+1}$ to the beginning of the list). Therefore, the total cost $A_2$ pays for all the events is $\sum_{\ell=1}^{n-1}\Theta(\ell)=\Theta(n^2)$. $A_2$ serves all the requests before (or at) their deadlines and thus it pays no penalty. As for $OPT$, it serves all the requests for $c_1$ together at time $2n$ by accessing $c_1$ (with a cost of $1$). The rest $n-1$ requests are not served and thus $OPT$ pays a penalty of $1+\epsilon$ for each one of them. $OPT$ does not perform any swap. Therefore $OPT(\sigma)=1+(n-1)\cdot(1+\epsilon)=\Theta(n)$ and thus $\frac{A_2(\sigma)}{OPT(\sigma)}=\frac{\Theta(n^2)}{\Theta(n)}=\Omega(n)$. Note that if the position of $c_1$ in $OPT$s list was not $1$ then $OPT$ would pay an access cost of at most $n$ when it accessed this element so we would still have $OPT(\sigma)=\Theta(n)$ and thus $\frac{A_2(\sigma)}{OPT(\sigma)}=\Omega(n)$.

\textbf{Algorithm $A_3$}: Same as Algorithm $A_2$ but all the element counters of the first $\ell$ elements in the list are set to $0$ - and not just $EC_e$.

\textbf{Counter example for Algorithm $A_3$}: Same as the counter example for $A_2$ but for each $\ell\in[n-1]$, the penalty $p_{2\ell+1}$ is changed from $1$ to $\ell+1-\epsilon$. $OPT$ behaves in the same way as before so $\frac{A_3(\sigma)}{OPT(\sigma)}=\frac{\Theta(n^2)}{n}=\Omega(n)$.

\textbf{Algorithm $A_4$}: Upon prefix element counter event on $\ell\in[n]$: Let $e$ be the element among the current first $\ell$ elements in the list which has the maximum value in its element counter. Serve the set of requests in the first $2\ell$ elements in the list, set ${EC}_e$ to $0$ and Move-to-front the element $e$.

\textbf{Counter example for Algorithm $A_4$}: The sequence $\sigma$ will contain $n$ requests: a request for each element in the list. The first $\sqrt{n}$ requests will be for the elements $c_{n-\sqrt{n}+1}$,$c_{n-\sqrt{n}+2}$,...$c_{n}$ (i.e. the last $\sqrt{n}$ elements in the list at time $t=0$). These requests will arrive at time $0$ with deadline $0$ and penalty of $\sqrt{n}-4$. Then the other $n-\sqrt{n}$ will arrive: the request for $c_\ell$ (for each $1\leq\ell\leq n-\sqrt{n}$) will arrive on time $n-\ell$, and this time will also be its deadline. Its penalty will be of $4\sqrt{n}$. For each $\ell\in[n-\sqrt{n}]$, at time $t=n-\ell$, $A_4$ will have a prefix element counter event on $n$ because of the request for the element $c_\ell$. The position of $c_\ell$ in the list at the beginning of this event will be $n-\sqrt{n}$ and this element will be the one to be moved to the beginning of the list. Therefore, $A_4$ pays access cost of $n$ and swapping cost of $n-\sqrt{n}-1$ for each of these $n-\sqrt{n}$ prefix element counter events. Therefore $A_4$ pays a total cost of $\Theta(n^2)$ in these events. Observe that $A_4$ pays delay penalty of $\sqrt{n}-4$ for each of the first $\sqrt{n}$ requests but the rest of the requests are served at their deadlines, thus $A_4$ does not pay their penalty. This adds $\Theta(n)$ to $A_4(\sigma)$, thus $A_4(\sigma)=\Theta(n^2)+\Theta(n)=\Theta(n^2)$. $OPT$ will not do anything and simply pay the penalty of $\Theta(\sqrt{n})$ for each of the $n$ requests, thus $OPT(\sigma)=\Theta(n\sqrt{n})$. Therefore $\frac{A_4(\sigma)}{OPT(\sigma)}=\frac{\Theta(n^2)}{\Theta(n\sqrt{n})}=\Omega(\sqrt{n})$.

\textbf{Algorithm $A_5$}: Same as Algorithm $A_4$ but all the element counters of the first $\ell$ elements in the list are set to $0$ - and not just $EC_e$.

\textbf{Counter example for Algorithm $A_5$}: The sequence $\sigma$ will contain $(n-\sqrt{n})\cdot(\sqrt{n}+1)$ requests: For each $\ell\in[n-\sqrt{n}]$, at time $n-\ell$, $\sqrt{n}+1$ requests will arrive together, all with deadline $n-\ell$: One request will be for the element $c_\ell$ with penalty of $4\sqrt{n}$ while the other $\sqrt{n}$ requests will be for the elements $c_{n-\sqrt{n}+1}$,$c_{n-\sqrt{n}+2}$,...$c_{n}$ and each one of them will have a penalty of $\sqrt{n}-4$. $A_5$ will act the same as $A_4$ acted, thus it will not suffer any penalty and its cost will be $A_5(\sigma)=\Theta(n^2)$. $OPT$ will initially move the elements $c_{n-\sqrt{n}+1}$,$c_{n-\sqrt{n}+2}$,...$c_{n}$ to the beginning of its list with a cost of $\Theta(n\sqrt{n})$ (The swapping cost needed to move each of the $\sqrt{n}$ elements to the beginning of $OPT$s list is $\Theta(n)$). Then for each $\ell\in[n-\sqrt{n}]$, at time $n-\ell$, $OPT$ will access the elements $c_{n-\sqrt{n}+1}$,$c_{n-\sqrt{n}+2}$,...$c_{n}$ (which are at the beginning of its list at this time) with a cost of $\sqrt{n}$, so the total access cost $OPT$ pays will be $\Theta(n\sqrt{n})$. Each request for the elements $c_{n-\sqrt{n}+1}$,$c_{n-\sqrt{n}+2}$,...$c_{n}$ is served at its deadline by $OPT$ and thus $OPT$ does not pay any penalty for it. $OPT$ pays the penalty of $4\sqrt{n}$ for each of the $n-\sqrt{n}$ requests for the other elements in the list, which $OPT$ does not serve, thus the total penalty $OPT$ pays is $\Theta(n\sqrt{n})$. To conclude, we have $OPT(\sigma)=\Theta(n\sqrt{n})$ and thus $\frac{A_5(\sigma)}{OPT(\sigma)}=\frac{\Theta(n^2)}{\Theta(n\sqrt{n})}=\Omega(\sqrt{n})$

\textbf{Algorithm $A_6$}: Upon prefix element counter event on $\ell\in[n]$: Let $M$ be the current maximum value of an element counter among the element counters for the first $\ell$ elements in the list. Among the first $\ell$ elements of the list, choose the most further element which has an element counter's value of at least $\frac{M}{2}$, let $e$ be this element. Serve the set of requests in the first $2\ell$ elements in the list, set ${EC}_e$ to $0$ and Move-to-front the element $e$.

\textbf{Counter example for Algorithm $A_6$}: The same counter example for Algorithm $A_4$ works here too.

\textbf{Algorithm $A_7$}: Same as Algorithm $A_6$ but all the element counters of the first $\ell$ elements in the list are set to $0$ - and not just $EC_e$.

\textbf{Counter example for Algorithm $A_7$}: The same counter example for Algorithm $A_5$ works here too.

\textbf{Algorithm $A_8$}: Upon prefix element counter event on $\ell\in[n]$: For each element $e^{'}$ among the first $\ell$ elements in the list, let $EC_{e^{'}}^{cur}$ be the current value of the element counter $EC_{e^{'}}$ and let $F_{e^{'}}^{cur}$ be the current sum of element counters for the elements which are currently located before $e^{'}$ in the list. Choose the element $e$ which maximizes the term $F_{e^{'}}^{cur}+2\cdot EC_{e^{'}}^{cur}$. Serve the set of requests in the first $2\ell$ elements in the list, set ${EC}_e$ to $0$ and Move-to-front the element $e$ \footnote{If we chose the element $e$ which maximizes the term $F_{e^{'}}^{cur}+ EC_{e^{'}}^{cur}$ instead of $F_{e^{'}}^{cur}+2\cdot EC_{e^{'}}^{cur}$ then we would choose the $\ell$-th element in the list currently and thus get an algorithm which is equivalent to Algorithm $A_2$. We act differently here than in Algorithm $A_2$ and one can observe that the counter example for Algorithm $A_2$ will not work on this algorithm.}.

\textbf{Counter example for Algorithm $A_8$}: The sequence $\sigma$ will contain $(n-\sqrt{n}+1)\cdot\sqrt{n}$ requests. For each $\sqrt{n}\leq\ell\leq n$, $\sqrt{n}$ requests will arrive at time $\ell$ and their deadline will be at time $\ell$. One request will be for the element $c_{\ell}$ and its penalty will be $\frac{\ell}{\sqrt{n}}$. The other $\sqrt{n}-1$ requests will be for the elements $c_1,c_2,...,c_{\sqrt{n}-1}$ and the penalty of each one of them will be $\frac{1}{\sqrt{n}}$. 
For each $\sqrt{n}\leq\ell\leq n$, a prefix element counter event on $\ell$ will occur at time $\ell$ and the element $c_{\ell}$ is the element that will be chosen by $A_8$ to be moved to the beginning of the list. $A_8$ pays an access cost of $\Theta(\ell)$ during this event for the access operation and for the move-to-front of $c_{\ell}$. $A_8$ does not pay any penalty because all the requests are served at their deadline. Therefore $A_8(\sigma)=\sum_{\ell=\sqrt{n}}^{n} \Theta(\ell)=\Theta(n^2)$. $OPT$ will serve all the requests for $c_1,c_2,...,c_{\sqrt{n}-1}$ at their deadlines (and thus it will not pay their penalties). This requires $OPT$ to pay $n-\sqrt{n}+1$ times an access cost of $\sqrt{n}-1$, thus a total access cost of $\Theta(n\sqrt{n})$. $OPT$ will have to pay the penalties for the requests for the elements $c_{\sqrt{n}},c_{\sqrt{n}+1},...,c_{n}$ thus the total penalty that $OPT$ pays is $\sum_{\ell=\sqrt{n}}^{n}\frac{\ell}{\sqrt{n}}=\Theta(n\sqrt{n})$. $OPT$ does not perform any swaps. Therefore $OPT(\sigma)=\Theta(n\sqrt{n})$ and thus $\frac{A_8(\sigma)}{OPT(\sigma)}=\frac{\Theta(n^2)}{\Theta(n\sqrt{n})}=\Omega(\sqrt{n})$.

\textbf{Algorithm $A_9$}: Same as Algorithm $A_8$ but all the element counters of the first $\ell$ elements in the list are set to $0$ - and not just $EC_e$.

\textbf{Counter example for Algorithm $A_9$}:
We make the following change to the sequence $\sigma$ from the counter example of Algorithm $A_8$. Recall that for each $\sqrt{n}\leq\ell\leq n$, there are requests for the elements $c_1,c_2,...,c_{\sqrt{n}-1}$ which arrive at time $\ell$ (and their deadline is also at time $\ell$). We change their penalties from $\frac{1}{\sqrt{n}}$ to $\frac{\ell}{\sqrt{n}}$, the same penalty as the request for $c_{\ell}$ which also arrives at time $\ell$ and also has a deadline of $\ell$. One can can observe that we have $A_9(\sigma)=\Theta(n^2)$. $OPT$ behaves the same as before and we have $OPT(\sigma)=\Theta(n\sqrt{n})$. To conclude, we have $\frac{A_9(\sigma)}{OPT(\sigma)}=\frac{\Theta(n^2)}{\Theta(n\sqrt{n})}=\Omega(\sqrt{n})$.
\end{document}